\documentclass[lettersize,journal]{IEEEtran}
    	\usepackage[top=0.7in, bottom=0.7in, left=0.6in, right=0.6in]{geometry}	
    	\usepackage{amsmath,amsfonts}
    	\usepackage{amsmath}
    	\usepackage{array}
    	\usepackage{textcomp}
    	\usepackage{stfloats}
    	\usepackage{url}
    	\usepackage{verbatim}
    	\usepackage{graphicx}
    	\usepackage{cite}
    	\usepackage{tikz}
    	\usepackage{amsthm}
    	\usepackage{subcaption}
    	\usepackage[normalem]{ulem}
    	\usepackage{etoolbox}
    	\usepackage{flushend}
    	\usepackage{amssymb}
    	\newtheorem{remark}{Remark}
    	
    	\usepackage{mathtools}
    	\newtheorem{proposition}{Proposition}
    	\usepackage{breqn}
    	\usepackage[multiple]{footmisc}
    	\newcommand{\widesim}[2][1.5]{
    		\mathrel{\overset{#2}{\scalebox{#1}[1]{$\sim$}}}
    	}
    	
    	\definecolor{ForestGreen}{RGB}{34,139,34}
    	\usepackage{bbm}
    	\newtheorem{theorem}{Theorem}
    	\newtheorem{lemma}{Lemma}
    	
    	\DeclareMathOperator*{\argmax}{arg\,max}
    	
    	\allowdisplaybreaks
    	
    	\usepackage{algorithm}
    	\usepackage{algpseudocode}
    	\usepackage{setspace}
    	\usepackage[font=footnotesize,labelfont=bf,skip=5pt]{caption}
    	\usepackage{subcaption}
    	\usepackage[inline,shortlabels]{enumitem}
    	\usepackage{bigints}
    	\usepackage{scalerel}[2016-12-29]
    	\def\stretchint#1{\vcenter{\hbox{\stretchto[440]{\displaystyle\int}{#1}}}}
    	\newcommand{\mybinom}[3][0.8]{\scalebox{#1}{$\dbinom{#2}{#3}$}}
    	\usepackage{wrapfig}
    	\usepackage[colorlinks]{hyperref}
    	\hypersetup{
    	colorlinks   = true, 
    	urlcolor     = blue,
    	linkcolor    = blue, 
    	filecolor	    = blue,
    	citecolor   = {blue}, 
    }

\begin{document}
    	\title{	On the Impact of an IRS on the Out-of-Band Performance in Sub-6 GHz \& mmWave Frequencies}
    		\author{L. Yashvanth,~\IEEEmembership{Student Member,~IEEE}, and Chandra R. Murthy,~\IEEEmembership{Fellow,~IEEE}
    			\thanks{This work was financially supported by the Qualcomm Innovation Fellowship, Qualcomm UR $6$G India Grant, a research grant from MeitY, Govt. of India, and the Prime Minister’s Research Fellowship, Govt. of India. A preliminary version of this work was presented at the 24th IEEE International Workshop on Signal Processing Advances in Wireless Communications (SPAWC), Sep. 2023 [DOI: 10.1109/SPAWC53906.2023.10304535] \cite{Yashvanth_SPAWC_2023}.}
    			\thanks{The authors are with the Dept. of ECE, Indian Institute of Science, Bengaluru, India 560 012. (e-mails: \{yashvanthl, cmurthy\}@iisc.ac.in).}
    		}
    		\maketitle
    		\begin{abstract}
    			Intelligent reflecting surfaces (IRSs) were introduced to enhance the performance of wireless communication systems. However, from a service provider's viewpoint, a concern with the use of an IRS is its effect on out-of-band (OOB) quality of service. Specifically, if two operators, say X and Y, provide services in a given geographical area using non-overlapping frequency bands, and if operator X uses an IRS to enhance the spectral efficiency (SE) of its users (UEs), does it degrade the performance of UEs served by operator~Y?
    			We answer this by analyzing the average and instantaneous performances of the OOB operator considering both sub-6 GHz and mmWave bands. Specifically, we derive the ergodic sum-SE achieved by the operators under round-robin scheduling. We also derive the outage probability and analyze the change in the SNR caused by the IRS at an OOB UE, using stochastic dominance theory.
    			Surprisingly, even though the IRS is randomly configured from operator Y's point of view, the OOB operator still benefits from the presence of the IRS, witnessing a performance enhancement \emph{for free} in both sub-6 GHz and mmWave bands. This is because the IRS introduces additional paths between the transmitter and receiver, increasing the overall signal power arriving at the UE and providing diversity benefits. Finally, we show that the use of opportunistic scheduling schemes can further enhance the benefit of the uncontrolled IRS at OOB UEs.
    			We numerically illustrate our findings and conclude that an IRS is always beneficial to every operator, even when the IRS is \textcolor{black}{deployed \& controlled by only one operator.\!}\! 
    		\end{abstract}
    		\begin{IEEEkeywords}
    			Intelligent reflecting surfaces, out-of-band performance, sub-6 GHz bands, mmWave communications.
    		\end{IEEEkeywords}
    		\section{Introduction}
    		Intelligent reflecting surfaces (IRSs) have been extensively studied to enhance the performance of beyond $5$G and $6$G communications~\cite{Basar_IA_2019,RuiZhang_IRSSurvey_TCOM_2021}. They are passive reflecting surfaces with many elements that can introduce carefully optimized phase shifts while reflecting an incoming radio-frequency (RF) signal, thereby steering the signal towards a desired user equipment (UE).  Although a recent topic of research, several use-cases elucidating the benefits of IRS-aided communications have already been investigated~\cite{RuiZhang_IRSSurvey_TCOM_2021,Alberto_TWC_2023,Xiaoming_TCOM_2023,Xiaowei_TVT_2023,Xiaowei_WCL_2023,Xiaowei_WCL_2022}. All these studies implicitly assume that only \emph{one operator} deploys and controls single or multiple IRS(s) to provide wireless services to their subscribed UEs. 
    		However, in practice, multiple wireless network operators co-exist in a given geographical area, each operating in non-overlapping frequency bands.  
		Then, if an IRS is deployed and optimized by one of the operators to cater to its subscribed UEs' needs, it is unclear whether the IRS will boost or degrade the performance of the UEs served by the other operators. This point is pertinent because the IRS elements are passive and have no band-pass filters. They reflect all the RF signals that impinge upon them across a wide range of frequency bands, including signals from \emph{out-of-band (OOB)} operators intended to be received by OOB UEs. Thus, this paper focuses on understanding how an IRS controlled by one operator affects the performance of other OOB operators.
    			\subsection{Related Literature and Novelty of This Work}
    			The benefits of using an optimized IRS has been illustrated in several use-cases, see~\cite{RuiZhang_IRSSurvey_TCOM_2021,Alberto_TWC_2023,Xiaoming_TCOM_2023,Xiaowei_TVT_2023,Xiaowei_WCL_2023,Xiaowei_WCL_2022}. For e.g., in~\cite{Alberto_TWC_2023}, outage analysis with randomly distributed IRSs is investigated. In~\cite{Xiaoming_TCOM_2023,Tabassum_TVT_2021}, coverage enhancement due to an IRS is studied, and in~\cite{Ekram_TCOM_2022}, the effect of IRSs on inter base station (BS) interferences is explored. Contrary to this, a few works also use randomly configured IRSs and obtain benefits from them. For instance,~\cite{Nadeem_WCL_2021,Yashvanth_TSP_2023,Nadeem_TWC_2021} use opportunistic communications using randomly configured IRSs. Similarly, blind beam forming approaches and diversity order analysis are reported in~\cite{Kaiming_TWC_2023_v1}, and~\cite{Constantinos_JSAC_2023,Psomas_TCOM_2021, Zhang_Random_IRS_WCL_2021}, respectively. Also, in~\cite{Liang_TWC_2024}, random IRSs are used to protect against wireless jammers. However, none of these works consider the scenario where multiple network operators coexist in an area. Very few studies have considered using an IRS in multi-band, multi-operator systems. In~\cite{Emil_Multiple_OP_TCOM_2024}, the effect of pilot contamination on channel estimation in a multiple-operator setting is reported. In~\cite{Cai_2022_TCOM,Cai_2020_CL_practicalIRS}, the authors jointly optimize the IRS configurations across multiple frequency bands via coordination among the BSs, which is impractical and incurs high signal processing overhead.\footnote{The works that assume inter-BS coordination to optimize the IRS configuration jointly consider a single operator deploying many BSs. Coordination among BSs may be feasible in this scenario. However, it is impractical for the BSs owned by two different operators to coordinate with each other.} Further, the solutions and analyses provided therein are not scalable with the number of operators (or frequency bands). More fundamentally, none of these works consider the effect on the OOB performance, even in the scenario of two operators providing services in non-overlapping bands when the IRS is optimized for only one operator. In this work, we investigate whether an IRS degrades the performance of other OOB operators. If not, can the mere presence of an IRS in the vicinity provide \emph{free gains} to OOB operators? Below, we explain our contributions in this context.
    			
    		\subsection{Contributions}\label{sec:contributions}
    		
    		We consider a system with two network operators, X and Y,  providing service in different frequency bands in the same geographical area. The IRS is optimized to serve the UEs subscribed to the \emph{in-band} operator X, and we are interested in analyzing the spectral efficiency (SE) achieved and outage probability witnessed by the UEs subscribed to the \emph{OOB} operator Y who does not control the IRS. Specifically, we (separately) evaluate the IRS-assisted performance in both sub-6 GHz and the mmWave bands, which are provisioned as the FR$1$ and FR$2$ bands in $5$G, respectively~\cite{3gpp_NR_frequency_bands}. Further, in the mmWave bands, inspired by~\cite{Wang_TWC_2022_IRS_BA}, we study two scenarios: (a) LoS (line-of-sight) and (b) (L+)NLoS (LoS and Non-LoS.) In the LoS scenario, the IRS is optimized or aligned to the dominant cascaded path (called the \emph{virtual} LoS path) of the in-band UE's channel. This is also considered in~\cite{Wang_TCOM_2023,Wang_TVT_2020}, where the in-band UEs' channels are approximated by the dominant LoS path to reduce the signaling overhead required for the base station (BS) to program the IRS: the phase of the second IRS element relative to the first element determines the entire phase configuration. Contrarily, in the (L+)NLoS case, the IRS optimally combines all the spatial paths to maximize the signal-to-noise ratio (SNR) at the receiver, for which the overhead scales linearly with the number of IRS elements. 
    Using tools from high-dimensional statistics, stochastic-dominance theory, and array processing theory, we make the following contributions.
    		 \subsubsection{OOB Performance in sub-6 GHz Bands (See Section~\ref{sec:OOB_perf_analysis})} 
    		Here, the operators serve their UEs over the sub-6 GHz frequency band where the channels are rich-scattering. In this context, our key findings are as follows. 		\begin{enumerate}[label=1-\alph*)]
    			\item  We derive the ergodic sum-SEs of the two operators as a function of the \textcolor{black}{system parameters}, under round-robin (RR) scheduling of the UEs served by both operators. We show that the sum-SE scales log-quadratically and log-linearly with the number of IRS elements for the in-band and OOB networks, respectively, even though the OOB operator does not control the IRS (see Theorem~\ref{thm:rate_characterization}.)\!
    			\item We show that the outage probability at an arbitrary OOB UE decreases monotonically with the number of IRS elements. Further, via the complementary cumulative distribution function (CCDF) of the difference in the OOB channel gain with and without the IRS, we prove that the OOB channel gain with an IRS \emph{stochastically dominates} the gain without the IRS, with the difference increasing with the number of IRS elements. Thus, an OOB UE \textcolor{black}{gets} instantaneous benefits that monotonically increase with the number of IRS elements (see Theorem~\ref{thm:exact_ccdf} and Proposition~\ref{sec:prop_stochastic_dominance_sub_6_GHz}.)\!\!
    			\end{enumerate}
    			\subsubsection{OOB Performance in mmWave Bands (See Section~\ref{sec:mmwave_LOS_RR})}
    			In the mmWave bands, the channels are directional, with only a few propagation paths. In this context, using novel probabilistic approaches based on the resolvable criteria of the mmWave spatial beams, our key findings are as follows. 
    			\begin{enumerate}[label=2-\alph*)]
    			\item In LoS scenarios, where the IRS is optimized to match the dominant path of the in-band UE's channel, we derive the ergodic sum-SEs of the two operators under RR scheduling of the UEs. The SE at the in-band UE scales log-quadratically in the number of IRS elements, whereas the SE gain at an OOB UE depends on the number of spatial paths in the OOB UE's channel. If there are a sufficient number of paths in the cascaded channel, the SE improvement due to the IRS scales log-linearly in the number of IRS elements. Otherwise, the OOB UE's SE improves only marginally compared to that in the absence of the IRS (see Theorem~\ref{thm:rate_characterization_mmwave_single_path_IB}.)
    			\item We evaluate the outage probability and CCDF of an OOB UE's channel gain with/without an IRS in LoS scenarios and prove that the channel gain in the presence of IRS \emph{stochastically dominates} the gain in its absence. Thus, even in mmWave bands, the IRS provides positive instantaneous gains to all the OOB UEs (see Theorem~\ref{thm:ccdf_mmwave_OOB_single_path}.)
    			\item We next consider the (L+)NLoS scenario where the IRS is jointly optimized considering all the spatial paths. We first evaluate the directional energy response of the IRS and show that it exhibits peaks only at the channel angles to which the IRS is optimized. This is a fundamental and new characterization of the IRS response when it is aligned to multiple paths in an mmWave system. (see Lemma~\ref{lemma_correlation_function_IRS_mmwave_multiple_path}.) 
    			\item We derive the ergodic sum-SE of both operators in (L+)NLoS scenarios. We find that the OOB performance is even better than the LoS scenario (and hence better than the system without an IRS.) This is because the odds that an OOB UE benefits improve when the IRS has a nonzero response in multiple directions. Thus, the OOB performance in mmWave bands does not degrade even when the OOB operator serves its UEs while remaining oblivious to the presence of the IRS (see Theorem~\ref{thm:ergodic_SE_mmwave_multiple_paths}.)
    			\end{enumerate}
    
    			\subsubsection{Opportunistic Enhancement of OOB Performance (See Section~\ref{sec:PF_MR_schedulers})} Having shown that an IRS positively benefits OOB UEs, we next suggest ways to exploit the uncontrolled IRS to enhance the performance of OOB operators further. In particular, by using opportunistic selection techniques, we leverage multi-user diversity and show that a significant boost in the OOB performance can be obtained compared to RR scheduling. Specifically, we demonstrate the following.
    			\begin{enumerate}[label=3-\alph*)]
    				\item By using a proportional-fair scheduler over a large number of OOB UEs, the sum-SE of operator Y converges to the so-called \emph{beamforming SE}, which is the SE obtained when the IRS is optimized for an OOB UE in every time slot (see Lemma~\ref{lem:PF_sub_6}.) 
    				\item By using a max-rate scheduler, the ergodic sum-SE of operator Y monotonically increases with both the number of IRS elements and OOB UEs in the system (see Lemma~\ref{lem:MR_scheduler_ergodic_SE}.)
    			\end{enumerate}			
    						Next, we highlight the practical utility of our results.
    
    \subsection{Practical Implications and Useful Insights}\label{sec:useful_insights}
    Our results offer several interesting insights into the performance of IRS-aided wireless systems where an IRS is deployed and controlled only by one operator. First, in all the scenarios mentioned in Sec.~\ref{sec:contributions}, we present novel, compact, and insightful analytical expressions that uncover the dependence of the performance on system parameters such as the number of IRS elements, SNR, number of channel paths, etc. As shown both analytically and through simulations, an IRS is beneficial to OOB users even though the IRS phase configuration is chosen randomly from the OOB operator's viewpoint, and this holds both in terms of the average and instantaneous SE. In particular, in rich scattering environments, the average SE at any OOB UE scales log-linearly in the number of IRS elements. On the other hand, in mmWave channels, the average SE improvement at the OOB UE due to the IRS is an increasing function of the number of spatial paths in the OOB UE's channel. Thus, deploying an IRS enriches the overall wireless channels (in both sub-6 GHz and mmWave bands) and can only benefit all wireless operators in the area.  Our study also reveals that there exists an interesting trade-off between reducing the signaling overhead to program the IRS at the in-band operator (as in LoS scenarios) versus boosting the gain at the OOB operators (as in (L+)NLoS scenarios) in the mmWave bands. Finally, our results and insights derived for $ 2$ operators directly extend to any number of OOB operators and to other settings like planar array IRS or multiple antenna BSs.

    \indent \textcolor{black}{\emph{Notation:} 
    $|\cdot|,\angle\cdot$ denote the magnitude and phase of a complex number (vector); $(\cdot)^*$ stands for complex conjugation; $\mathbbm{1}_{\{\cdot\}}$ is the indicator function; $|\mathcal{A}|$ is the cardinality of set $\mathcal{A}$; $A \stackrel{d}{=} B$ means $A$ and $B$ have the same distribution; $\odot$ is the Hadamard product; $\Re(\cdot)$ and $\mathbb{R}^+$ are the real part and the set of positive real numbers; ${ \sf{Pr}}(\cdot)$ and $\mathbb{E}[\cdot]$ refer to the probability measure and expectation operator. $\mathcal{O}(\cdot), \Omega(\cdot)$, and $o(\cdot)$ denote Landau's Big-O, Big-Omega, and little-o functions.} \!\!
    \vspace{-0.05cm}
    
    \section{System Model}\label{sec:system_model}
    		We consider a single-cell system with two mobile network operators, X and Y, who provide service to $K$ and $Q$ UEs, respectively, on different frequency bands and in the same geographical area. The BSs of operators X and Y (referred to as BS-X and BS-Y, respectively) and UEs are equipped with a single antenna, and all the channels undergo frequency-flat fading\textcolor{black}{\cite{Basar_TCOM_2022_iid}}.\footnote{For simplicity of exposition and to focus on the \emph{effect of the IRS on OOB users}, we consider single antennas at the BSs and UEs, similar to~\cite{Muhammad_Ali_TWC_2022,Xiaowei_TVT_2023,Xiaowei_WCL_2023,Xiaowei_WCL_2022,Nasiri-Kenari_WCL_2022}. The extension to the multiple antenna case does not change our broad conclusions, but we relegate this to future work.}	
    		An $N$-element IRS is deployed by operator X to enhance the SNR at the UEs it serves. So, operator X configures the IRS with the SNR-optimal phase configuration for a UE scheduled by BS-X in every time slot. On the other hand, operator Y does not deploy any IRS and is oblivious to the presence of operator X's IRS. 
    			\begin{figure}[t!]
			\vspace{-0.2cm}
    			\centering
    			\includegraphics[scale=0.15]{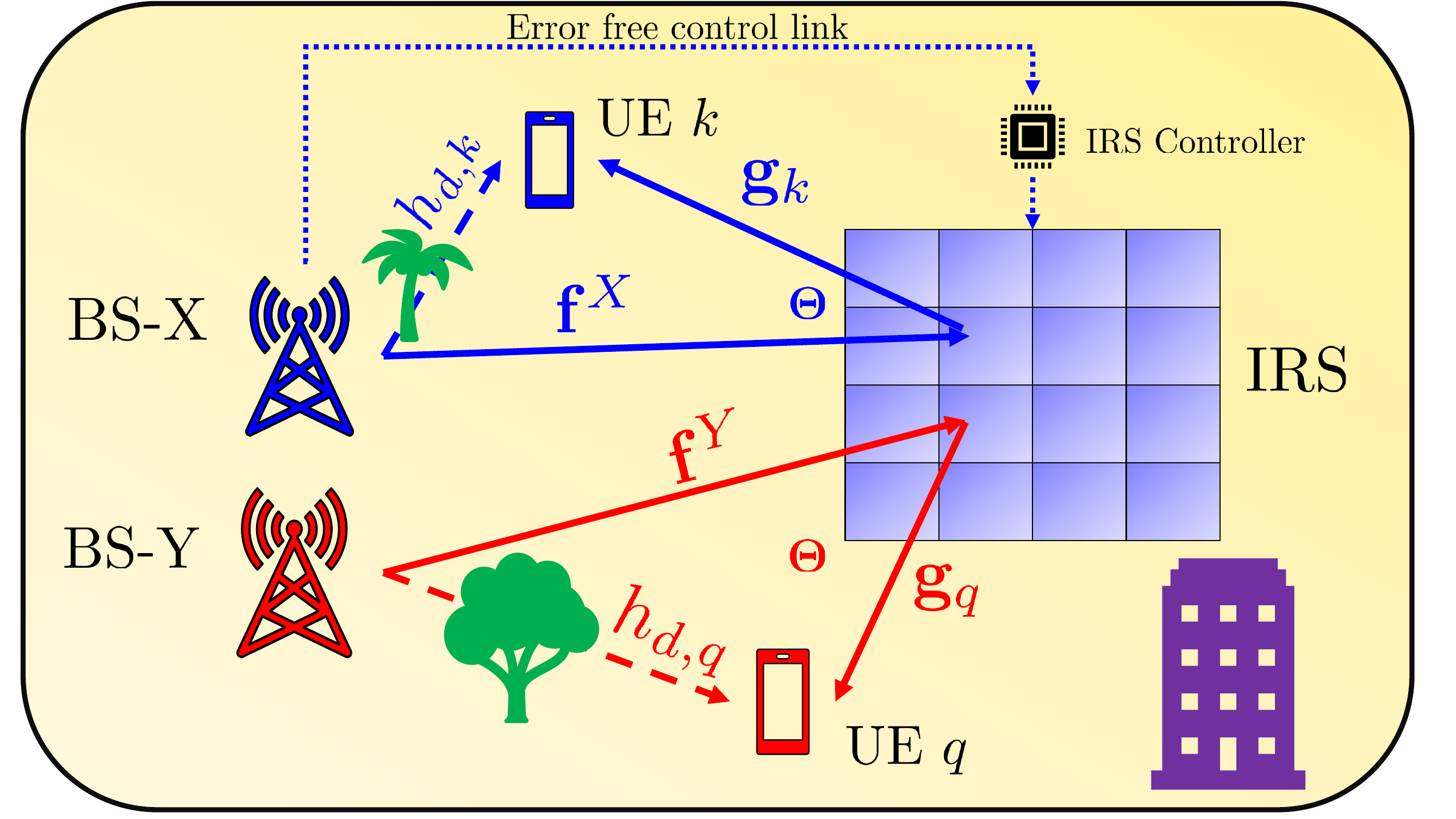}
    			\caption{Network scenario of an IRS-aided two-operator system.}
    			\label{fig:Network_scenario_single_IRS}
    			\vspace{-0.2cm}
    		\end{figure}
    		Then, the downlink signal received at the $k$th UE served by BS-X is\!
    		\begin{equation}\label{eq:airtel_downlink}
    			y_{k} = \left(h_{d,k}+\mathbf{g}_{k}^T\boldsymbol{\Theta}\mathbf{f}^X\right)x_k + n_k,
    		\end{equation} where $\mathbf{g}_k \in \mathbb{C}^{N \times 1}$ is the channel from IRS to the $k$th UE, $\mathbf{f}^X \in \mathbb{C}^{N \times 1}$ is the channel from BS-X to the IRS,  $\boldsymbol{\Theta} \in \mathbb{C}^{N \times N}$ is a diagonal matrix containing the IRS reflection coefficients of the form $e^{j\theta}$, and $h_{d,k}$ is the direct (non-IRS) path from BS to the UE-$k$. Also, $x_k$ is the data symbol for UE-$k$ with average power $\mathbb{E}[|x_k|^2] = P$, and $n_k$ is the additive noise $\sim \mathcal{CN}(0,\sigma^2)$ at UE-$k$. Similarly, at UE-$q$ served by BS-Y, we have
    		\begin{equation}\label{eq:jio_downlink}
    			y_{q} = \left(h_{d,q}+\mathbf{g}_{q}^T\boldsymbol{\Theta}\mathbf{f}^Y\right)x_q	+ n_q.
    		\end{equation} 
    	In Figure~\ref{fig:Network_scenario_single_IRS}, we pictorially illustrate the considered network.
    		\vspace{-0.3cm}
    		\subsection{Channel Model in sub-6 GHz Frequency Bands}\label{sec:ch_model_sub6_GHz}
    		Similar to~\cite{Basar_TCOM_2022_iid}, in the sub-6 GHz band, we consider that all the fading channels are statistically independent\footnote{We consider that IRS elements are placed sufficiently far apart that the spatial correlation in the channels to/from the IRS is negligible.} and identically distributed (i.i.d.) following the Rayleigh distribution.\footnote{Our results can be easily extended to other sub-$6$ GHz channel fading models also. From the approach used to show the results here, it is easy to see that the scaling law of the SE as a function of the IRS parameters will remain the same, although the scaling constants depend on the fading model. Thus, the conclusions on the impact of the IRS on OOB performance reported in this paper will continue to hold.} Specifically, $h_{d,k} = \sqrt{\beta_{d,k}} \tilde{h}_{d,k}, h_{d,q} = \sqrt{\beta_{d,q}} \tilde{h}_{d,q}; \tilde{h}_{d,k}, \tilde{h}_{d,q}  \widesim[2]{\text{i.i.d.}} \mathcal{CN}(0,1)$; $\mathbf{g}_k = \sqrt{\beta_{\mathbf{g},k}} \tilde{\mathbf{g}}_k, \mathbf{g}_q = \sqrt{\beta_{\mathbf{g},q}} \tilde{\mathbf{g}}_q; \tilde{\mathbf{g}}_k,\tilde{\mathbf{g}}_q \widesim[2]{\text{i.i.d.}} \mathcal{CN}(\mathbf{0}, \mathbf{I}_N)$; $\mathbf{f}^X = \sqrt{\beta_{\mathbf{f}^X}} \tilde{\mathbf{f}}^X, \mathbf{f}^Y = \sqrt{\beta_{\mathbf{f}^Y}} \tilde{\mathbf{f}}^Y; \tilde{\mathbf{f}}^X,\tilde{\mathbf{f}}^Y \widesim[2]{\text{i.i.d.}} \mathcal{CN}(\mathbf{0}, \mathbf{I}_N)$. All terms of the form $\beta_x$ represent the path losses. \!\!			\vspace{-0.3cm}
    		\subsection{Channel Model in mmWave Frequency Bands}\label{sec:ch_model_mmwave}
    		In the mmWave band, we consider a Saleh-Valenzuela type model for all the channels~\cite{Heath_TWC_2014}. At UE-$\ell$, the channels are 
    		\begin{equation}\label{eq:ch_model_mmwave_single_IRS}
    			\mathbf{f}^{p} \!=\!\! \sqrt{\!\frac{N}{L_{1,p}}}\!\sum_{i=1}^{L_{1,p}}\!\!\gamma^{(1)}_{i,p} \!\mathbf{a}^*_N(\phi_{i,p});\!\text{  } \mathbf{g}_{\ell} \!=\!\! \sqrt{\!\frac{N}{L_{2,\ell}}}\!\sum_{j=1}^{L_{2,\ell}}\!\!\gamma^{(2)}_{j,\ell} \!\mathbf{a}^*_N(\psi_{j,\ell}),\!\!\!
    		\end{equation} where $p \in \{X,Y\}$, $L_{1,p}$ and $L_{2,\ell}$ are the number of spatial paths in the BS-$p$ to IRS, and IRS to UE-$\ell$ links, respectively. Since the in-band and OOB BSs/UEs are distributed arbitrarily with respect to the IRS, for notational simplicity, we let $L_{1,X}=L_{1,Y}\triangleq L_1$, and $L_{2,\ell}= L_2 $ $\forall \ell$ (including in-band and OOB UEs.) Also, $\phi_{i,p}$ and $\psi_{j,\ell}$ denote the sine of the angle of arrival of the signal from BS-$p$ to the IRS via the $i$th path, and the sine of the angle of departure from the IRS to the $\ell$th UE via the $j$th path, where sine of an angle ($\phi_x$) is related to the physical spatial angle ($\chi$) by\footnote{In the sequel, the term ``angle" will denote the sine of a physical angle.}
    	$\phi_x =  (2d/\lambda)\sin(\chi)$, with $d$, $\lambda$ being the inter-elemental distance and signal wavelength, respectively.
        	The sine terms are sampled from a distribution $\mathcal{U}_{\mathcal{A}}$, which depends on the beam resolution capability of the IRS (elaborated next.) The fading coefficients, $\gamma^{(1)}_{i,p}$ and $\gamma^{(2)}_{j,\ell}$, are independently sampled from $\mathcal{CN}(0,\beta_{\mathbf{f}^p})$, and $\mathcal{CN}(0,\beta_{\mathbf{g},\ell})$, respectively. Finally, $\mathbf{a}_N(\phi)$ is an array steering vector of a uniform linear array (ULA)\footnote{All our main conclusions remain unchanged even if other array geometries (e.g., a uniform planar array) are used at the IRS, similar to~\cite{Wei_Yu_TWC_2023_Sparse_Ch_estim_IRS,Lee_WCL_2023,Muhammad_Ali_TCOM_2022,Nasiri-Kenari_Arxiv_2022}.} 
        	based IRS oriented at the angle $\phi$,  given by 
        		\begin{equation}\label{eq:array_vector_template}
        			\mathbf{a}_N(\phi) =\frac{1}{\sqrt{N}} [1, e^{-j\pi\phi},\ldots,e^{-j(N-1)\pi\phi}]^T.
        		\end{equation}
    
    	\subsubsection{Beam Resolution Capability of the IRS}
    	We now describe the beam resolution capability of an IRS, which is important to account for when the array contains a finite number of elements. Beam resolution measures the degree to which two closely located beams are distinguishable (or resolvable.) Spatial paths  in~\eqref{eq:ch_model_mmwave_single_IRS} will become non-resolvable if the angles associated with the paths are spaced less the beam resolution and can be modeled as being clustered into a single path with an appropriate channel gain. To measure the beam resolution of the ULA-based IRS model~\cite{Wei_Yu_TWC_2023_Sparse_Ch_estim_IRS,Lee_WCL_2023,Muhammad_Ali_TCOM_2022,Nasiri-Kenari_Arxiv_2022}, we enumerate all possible $D$ beams as
    	\begin{equation}
    		\mathbf{A} = \left[\mathbf{a}_N(\phi_1),\mathbf{a}_N(\phi_2), \ldots,\mathbf{a}_N(\phi_D) \right] \in \mathbb{C}^{N\times D}.
    	\end{equation}
    	We can verify that $\mathbf{A}$ is a Vandermonde matrix and hence its column vectors $\left\{ \mathbf{a}_N(\phi_i)\right\}^{D}_{i=1}$ are linearly independent when $D\!\leq\! N$ provided $\{\phi_i\}_{i=1}^D$  are distinct~\cite[Pg. 185]{Meyer_Matrix_book_2000}. On the other hand, when $D>N$, the vectors are always linearly dependent and hence are non-resolvable. So, the total number of independent (or resolvable) beams that can be formed by an $N$-element IRS is at most $N$. This fact has been observed in the   literature~\cite{Heath_TWC_2014,Gen_TWC_2016,Vasanthan_JSTSP_2007}, and also in the current works on IRS-aided mmWave wireless systems~\cite{Wang_TWC_2022_IRS_BA,Rui_Zhang_WCL_2020}. Thus, the complete set of resolvable beams $\mathcal{A}$ (called the \emph{resolvable beambook}) at the IRS is\footnote{For analytical tractability, we consider a \emph{flat-top} RF directivity pattern of the IRS~\cite[Eq. 5]{Madhow_TACM_2011}. This has been shown to be a good approximation to practical array systems and becomes accurate as $N$ gets large~\cite{Madhow_TACM_2011,Heath_TWC_2014}.}
    	\begin{equation}\label{eq:resolvable_paths}
    		\mathcal{A} \triangleq \left\{ \mathbf{a}_N(\phi), \phi \in \mathbf{\Phi}\right\}; \mathbf{\Phi} \triangleq\! \left\{\! \left(\!-1 + \!\dfrac{2i}{N}\right)\!\bigg\rvert i =  0,\ldots,N-1\!\right\}\!.\!
    	\end{equation} 
	Here $\mathbf{\Phi}$ is the \emph{resolvable anglebook} of the IRS, and, without loss in generality, we model its distribution $\mathcal{U}_{\mathcal{A}}$ by
    	\begin{equation}\label{eq:codeboook_ditbn}
    		\mathcal{U}_{\mathcal{A}}(\phi) = \frac{1}{\left| \mathbf{\Phi}\right|}   \mathbbm{1}_{\left\{\phi \in \mathbf{\Phi}\right\}} = \frac{1}{N}   \mathbbm{1}_{\left\{\phi \in \mathbf{\Phi}\right\}}.
    	\end{equation} 
    	Also, for large $N$, $\mathcal{A}$ constitutes a set of orthonormal vectors:\!\!
    	\begin{equation}\label{eq:orthogonal_array_vectors}
    	\mathbf{a}^H_N(\phi_1)\mathbf{a}_N(\phi_2) \rightarrow \delta_{\{\phi_1,\phi_2\}} \hspace{0.5cm} \forall \hspace{0.2cm}\mathbf{a}_N(\phi_1),\mathbf{a}_N(\phi_2) \in \mathcal{A} ,
    	\end{equation} where $\delta_{\{x,y\}}$ is the usual dirac-delta function. 
    	\subsection{Problem Statement} We are now ready to state the problem mathematically. Suppose BS-X tunes the IRS with the SNR-optimal vector $\boldsymbol{\theta}^{\mathrm{opt}} \triangleq \mathrm{diag}(\boldsymbol{\Theta}^{\mathrm{opt}})$ to serve UE-$k$ by solving the problem:
    \begin{equation}	
    \boldsymbol{\Theta}^{\mathrm{opt}}=\argmax_{\boldsymbol{\Theta}\in \mathbb{C}^{N \times N}} \quad  \left|h_{d,k}+\mathbf{g}_{k}^T\boldsymbol{\Theta}\mathbf{f}^{X} \right|^2,	
    \end{equation}
    subject to $\boldsymbol{\Theta}$ being a diagonal matrix with unit-magnitude diagonal entries.
    Then, we wish to characterize the impact of the IRS on the channel gain, $\left|h_{d,q}\!+\mathbf{g}_{q}^T\boldsymbol{\Theta}^{\mathrm{opt}}\mathbf{f}^Y \right|^2$, of OOB UE-$q$ scheduled by BS-Y. In particular, we wish to answer:\!\!
    \begin{enumerate}
    \item Does the IRS degrade OOB performance? That is, when or how often will $\left|h_{d,q}+\mathbf{g}_{q}^T\boldsymbol{\Theta}^{\mathrm{opt}}\mathbf{f}^Y \right|^2 \!< \! \left|h_{d,q}\right|^2$ hold?
    \item How does the OOB channel gain $\left|h_{d,q}+\mathbf{g}_{q}^T\boldsymbol{\Theta}^{\mathrm{opt}}\mathbf{f}^Y \right|^2$ depend on $N$?
    \end{enumerate}
    Clearly, these aspects of an IRS are fundamental to understanding the overall impact of an IRS in practical systems when many operators co-exist. In the sequel, we answer these by analyzing the ergodic and instantaneous characteristics of the OOB UE's channel in both sub-6 GHz and mmWave bands in the presence of an IRS controlled by a different operator.
    \begin{remark}
      The next two sections focus on cases where both operators provide services in the sub-$6$ GHz band (Sec.~\ref{sec:OOB_perf_analysis}) or both provide services in the mmWave band (Sec.~\ref{sec:mmwave_LOS_RR}), using non-overlapping frequency allocations. We do not consider the case where one operator uses the sub-$6$ GHz band while the other uses the mmWave band, as it is unclear whether an IRS can efficiently reflect signals in both these frequency bands.
    \end{remark}
    		\section{OOB Performance: sub-6 GHz bands} \label{sec:OOB_perf_analysis}
    		Suppose an operator X deploys and controls an IRS to enhance the SE of the UEs being served by it in a sub-6 GHz band. We wish to characterize the effect of the IRS on operator Y, which is operating in a different sub-6 GHz frequency band with no control over the IRS. Thus, to serve the $k$th UE, BS-X configures the IRS with the SNR/SE-optimal phases~\cite{Basar_IA_2019,RuiZhang_IRSSurvey_TCOM_2021,Yashvanth_TSP_2023} 
    	    		\begin{equation}
    			\theta ^{\mathrm{opt}}_{n,k} =e^{j\left(\angle h_{d,k}- \left(\angle {f}^X_{n}+\angle {g}_{k,n}\right)\right)},\hspace{0.5cm} n=1, {\dots }, N, \label{eq:basic_optimal_angle_X}  
    		\end{equation}
    		which results in the coherent addition of the signals along the direct and IRS paths, leading to the maximum possible received SNR. Then, the SE achieved by the $k$th UE is 
    		\begin{equation} \label{eq_BF_rate_UE_k}
    			R_{k}^{BF} = \log _{2}\!\left( \! 1 + \frac {P}{\sigma ^{2}} \left| |h_{d,k}| + \sum\nolimits_{n=1}^{N} |f^X_{n}{g_{k,n}}| \right|^{2}\right).
    		\end{equation}
    		Due to the independence of the channels of the UEs served by operators X and Y, the IRS phase configuration used by operator X to serve its own UEs appears as a \emph{random} phase configuration of the IRS for any UE served by operator~Y. 
    		
    		We consider RR scheduling of UEs at both BS-X and BS-Y.\footnote{The extension of our results to the proportional fair and max-rate scheduling schemes is provided in Sec.~\ref{sec:PF_MR_schedulers}.} Since the BSs are equipped with a single antenna, one UE from each network is scheduled in every time slot. A summary of the protocol is given in Fig.~\ref{fig:flowchart}.
		    		 
    		 We characterize the average OOB performance by deriving the ergodic sum-SE of both networks and then infer the degree of degradation/enhancement of the OOB performance caused by the IRS. The ergodic SE at UE-$k$ is 
    		\begin{equation}\label{eq:basic_airtel_rate}
    			\!\langle R_k^{(X)} \rangle = \mathbb{E}\!\left[\log_{2}\left(\!1 \!+ \left|\sum\nolimits_{n=1}^N  \left|f^X_{n}{g_{k,n}}\right| 
    			\!+ \!  \left|h_{d,k}\right| \right|^2 \!\!\!\dfrac{P}{\sigma^2} \right)\!\!\right]\!,\!\!\!\!
    		\end{equation} since the IRS is configured with the optimal phase configuration for (scheduled) UE-$k$.
    		On the other hand, the ergodic SE for (scheduled) UE-$q$ of operator Y is 
    		\begin{equation}\label{eq:basic_jio_rate}
    			\langle R_q^{(Y)} \rangle = \mathbb{E}\!\left[\log_{2}\!\left(\!1 \!+ \left|\sum\nolimits_{n=1}^N \! f^Y_{n}{g_{q,n}} 
    		\!	+   h_{d,q} \right|^2\!\! \dfrac{P}{\sigma^2} \right)\right],\!\!
    		\end{equation} where we used the fact that the channels are circularly symmetric random variables, i.e.,  $f^Y_{n}g_{q,n}e^{j\theta} \stackrel{d}{=} f^Y_{n}g_{q,n}$ for any $\theta$.
    	    	Here, the expectations are taken with respect to the distribution of the channels to the respective UEs. With RR scheduling, the ergodic sum-SEs of two operators are given by
    		\begin{equation}\label{eq:sum-rate-template}
    			\bar{R}^{(X)}  \triangleq \frac{1}{K} \sum_{k=1}^K \langle R_k^{(X)} \rangle, \text{  and  }
    			\bar{R}^{(Y)} \triangleq \frac{1}{Q}\sum_{q=1}^Q \langle R_q^{(Y)} \rangle.
    		\end{equation}
    		We note that closed-form expressions for the ergodic sum-SE are difficult to obtain due to the complicated distribution of the SNR and SE terms (e.g., in~\eqref{eq:basic_airtel_rate} and~\eqref{eq:basic_jio_rate}.) Although we can use the approach in~\cite{Boulogeorgos_5GWF_2020} to obtain the exact ergodic sum-SE in terms of standard integrals, the resulting expressions do not provide insights into how the sum-SE scales with the system parameters. 
    		    		 Instead, we rely on applying tight approximations to obtain insightful results. However, we note that it is important to choose approximations such that the analysis across different scenarios can be unified in a single framework and the corresponding results are comparable.  To that end, we apply Jensen's inequality, and after careful simplification, we arrive at elegant and interpretable expressions for the ergodic sum-SEs in all scenarios. In this view, we have the following theorem for the sub-6 GHz band of communication. \!\!
    			\begin{figure}
    				\vspace{-0.1cm}
    				\centering
    			\includegraphics[width=0.9\linewidth]{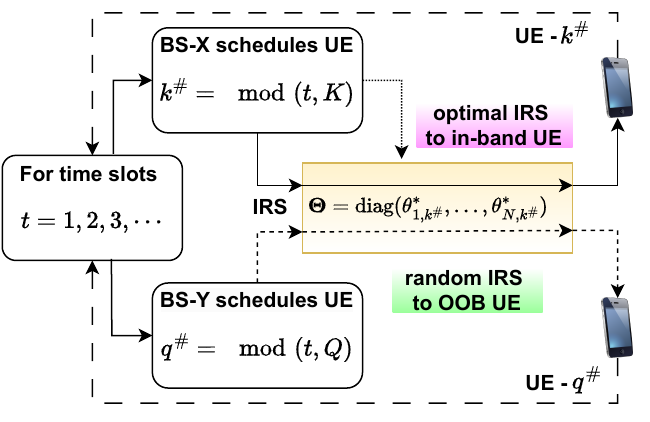}
    			\caption{Flowchart of the round-robin scheduling based protocol.}
    			\label{fig:flowchart}
    			\vspace{-0.3cm}
    		\end{figure}
    \begin{theorem}\label{thm:rate_characterization}
    Under independent Rayleigh fading channels in the sub-6 GHz bands,  with RR scheduling, and when the IRS is optimized to serve the UEs of operator X, the ergodic sum-SEs of operators X and Y  scale~as
    			\begin{multline}\label{eq:rate_airtel_rr}
    				\bar{R}^{(X)}\approx \frac{1}{K}\!\sum_{k=1}^K \log_{2}\left(1 +  \left[N^2\!\left(\frac{\pi^2}{16} \beta_{r,k}\!\right)\! + \right.\right.\\ \left.\left.\! N\!\left(\!\beta_{r,k}\!-\frac{\pi^2}{16}\beta_{r,k}+\frac{\pi^{3/2}}{4}\sqrt{\beta_{d,k}\beta_{r,k}}\right)\!+\!\beta_{d,k}\right]\!\frac{P}{\sigma^2}\right)\! ,\!
    			\end{multline}
    			 where $\beta_{r,k} \triangleq \beta_{\mathbf{f}^X}\beta_{\mathbf{g},k}$, and
    			\begin{equation}\label{eq:rate_jio_rr}
    				\bar{R}^{(Y)} \approx \frac{1}{Q}\sum_{q=1}^Q  \log_{2}\left(1 + \left[N\beta_{r,q} + \beta_{d,q}\right]\frac{P}{\sigma^2} \right),
    			\end{equation} 
    				where $\beta_{r,q} \triangleq \beta_{\mathbf{f}^Y}\beta_{\mathbf{g},q}$.
    		\end{theorem}
    		\begin{proof}
    			We compute the sum-SEs for both operators below. 
    			\subsubsection{Ergodic sum-SE of operator X}	 
    			We first compute $\langle R_k^{(X)} \rangle$ for a given $k$. By Jensen's inequality, we obtain
    			\begin{equation}\label{eq:airtel_jensen}
    				\!\!\langle R_k^{(X)} \rangle \leq	
    				\log_{2}\!\left(\!\!1 \!+\! \mathbb{E}\left[\left|\sum\nolimits_{n=1}^N  |f^X_{n}{g_{k,n}}| +   |h_{d,k}| \right|^2\right]\! \frac{P}{\sigma^2}\! \right)\!.\!
    			\end{equation}
    			We expand the expectation term as follows.
    			\begin{multline}\label{eq:op_X_mean_expand}
    				\!\!\!\!\left|\sum_{n=1}^N  |f^X_{n}{g_{k,n}}| 
    				+   |h_{d,k}| \right|^2\!\! = \sum_{n,m=1}^N |f^X_{n}||g_{k,n}||f^X_{m}||g_{k,m}|+ |h_{d,k}|^2 \\ + 2\sum\nolimits_{n=1}^N |f^X_{n}||g_{k,n}||h_{d,k}| \ \
    				= \sum\nolimits_{n=1}^{N}|f^X_{n}|^2|g_{k,n}|^2   + |h_{d,k}|^2 + \\  \mathop{\sum\nolimits_{n,m=1}^{N}}_{n \neq m} |f^X_{n}||g_{k,n}||f^X_{m}||g_{k,m}|+  2\sum_{n=1}^N |f^X_{n}||g_{k,n}||h_{d,k}|.
    			\end{multline}
    			Under Rayleigh fading, $\mathbb{E}[|f_n^X|^2] = \beta_{\mathbf{f}^X}, \mathbb{E}[|g_{k,n}|^2] = \beta_{\mathbf{g},k}, \mathbb{E}[|h_{d,k}|^2] = \beta_{d,k}$, $ \mathbb{E}[|f_n^X|] = \sqrt{\dfrac{\pi}{4}\beta_{\mathbf{f}^X}},$ $ \mathbb{E}[|g_{k,n}|] = \sqrt{\dfrac{\pi}{4}\beta_{\mathbf{g},k}}, \mathbb{E}[|h_{d,k}|]   =\! \sqrt{\dfrac{\pi}{4}\beta_{d,k}}, \forall  k =1,\ldots,K; n = 1,\ldots,N$. Further, all the random variables are independent. Taking the expectation in~\eqref{eq:op_X_mean_expand}, and substituting for these values, we get 
    			\begin{multline}\label{eq:expectation_iid_airtel}
    				\mathbb{E}\left[\left|\sum\nolimits_{n=1}^N  |f^X_{n}{g_{k,n}}| 
    				+   |h_{d,k}| \right|^2\right] = 
    				N^2\left(\frac{\pi^2}{16} \beta_{r,k}\right) \\+ N\left(\beta_{r,k}-\frac{\pi^2}{16}\beta_{r,k}+\frac{\pi^{3/2}}{4}\sqrt{\beta_{d,k}\beta_{r,k}}\right) +\beta_{d,k},
    			\end{multline}
    			where $\beta_{r,k}$ is as defined in the statement of the theorem. 		Substituting~\eqref{eq:expectation_iid_airtel} in~\eqref{eq:airtel_jensen}, and plugging in the resulting expression in~\eqref{eq:sum-rate-template} yields~\eqref{eq:rate_airtel_rr}, as desired.
    			\subsubsection{Ergodic sum-SE of operator Y}\label{proof:sub-6-jio-ergodi-rate}	 
    			As above, from Jensen's inequality, we have
    			\begin{equation}\label{eq:jio_jensen}
    				\!\!\langle R_q^{(Y)} \rangle \leq \log_{2}\left(\!1 + \mathbb{E}\left[\left|\sum\nolimits_{n=1}^N  f^Y_{n}{g_{q,n}} 
    					+   h_{d,q} \right|^2\right]\! \frac{P}{\sigma^2} \right)\!.\!\!
    			\end{equation}
    			Proceeding along the same lines, we get
    			\begin{multline}
    				\left|\sum\nolimits_{n=1}^N  f^Y_{n}{g_{q,n}} 
    				+  h_{d,q} \right|^2 
    				= \sum\nolimits_{n=1}^{N} |f^Y_{n}|^2|g_{q,n}|^2  +  |h_{d,q}|^2 + \\ \!
    				\mathop{\sum\nolimits_{n,m=1}^{N}}_{n \neq m} f^Y_{n}g_{q,n}f^{Y*}_{m}g^*_{q,m}   \!+\! 2\Re \left(\!\sum\nolimits_{n=1}^N f^Y_{n}g_{q,n}h^*_{d,q}\right) .
    			\end{multline}
    			Taking the expectation and simplifying,
    			\begin{equation}\label{eq:expectation_iid_jio}
    				\mathbb{E}\left[\left|\sum\nolimits_{n=1}^N  f^Y_{n}{g_{q,n}} 
    				+   h_{d,q} \right|^2\right] = N\beta_{r,q} + \beta_{d,k}.
    			\end{equation}
    			Substituting~\eqref{eq:expectation_iid_jio}  in~\eqref{eq:jio_jensen}, and plugging in the resulting expression in~\eqref{eq:sum-rate-template} yields~\eqref{eq:rate_jio_rr}.	 
    		\end{proof}
    \vspace{-0.2cm}
    		From the above theorem, we infer the following: 
		    		\begin{itemize}
    			\item The IRS enhances the average received SNR by a factor of $N^2$ at any scheduled (in-band) UE of operator X when BS-X optimizes the IRS. This is the benefit that operator X obtains by using an optimized $N$-element IRS.\!
    			\item Operator Y, who does not control the IRS, also witnesses an enhancement of average SNR by a factor of $N$ for free, i.e., without any coordination with the IRS. This happens because the IRS makes the wireless environment more rich-scattering on average, and facilitates the reception of multiple copies of the signals at the (OOB) UEs.
    		\end{itemize}
    		Next, we prove that even the instantaneous characteristics of the OOB channel are favorable due to the IRS. We recognize that the instantaneous channel at an OOB UE-$q$ is given by $|h_q|^2$, where $h_q \triangleq \sum_{n=1}^N  f^Y_{n}{g_{q,n}} 	+   h_{d,q} $. In the sequel, we provide two kinds of results: first, we compute the outage probability experienced by UE-$q$ and then provide stronger results on the OOB channels via stochastic dominance theory. 
    		
    		The outage probability of UE-$q$ is given by 
    		\begin{equation}\label{eq:outage_basic_expression_sub_6}
    			P^\rho_{q,\text{out}} = \mathsf{Pr}(|h_{q}|^2 < \rho), 
    		\end{equation} where $\rho$ is a constant that depends on the receiver sensitivity threshold.
    		Although works like~\cite{Psomas_TCOM_2021} provide closed-form expressions for the outage probabilities of randomly configured IRSs, these expressions do not provide insight into the system performance. So, using~\cite[Proposition~$1$]{Yashvanth_TSP_2023}, we approximate $h_{q} \sim \mathcal{CN}(0,N\beta_{r,q}+\beta_{d,q})$ which becomes accurate as $N$ gets large.\footnote{This approximation works well even with $N=8$~\cite{Yashvanth_TSP_2023}.} Hence, $|h_{q}|^2$ is a exponential random variable with mean $N\beta_{r,q}+\beta_{d,q}$, and the outage probability in~\eqref{eq:outage_basic_expression_sub_6} can be easily obtained as shown in Theorem~\ref{thm:exact_ccdf} below.
    		
    		Next, we characterize the stochastic behavior of the channel gains witnessed by UE-$q$ with/without an IRS. In this view, define the following random variables.
    		\begin{equation}\label{eq:aux_RV_defn}
    			\left|h_{1,q}\right|^2\triangleq \left|\sum_{n=1}^N  f^Y_{n}{g_{q,n}} 	+   h_{d,q} \right|^2;  \left|h_{2,q}\right|^2 \triangleq \left| h_{d,q} \right|^2.
    		\end{equation} 
    		Note that $\left|h_{1,q}\right|^2$ and $\left|h_{2,q}\right|^2$ represent the channel power gain of UE-$q$ in the presence and absence of the IRS, respectively. We  now characterize the \emph{change} in the channel gain at UE-$q$ served by BS-Y in the presence and absence of the IRS:
    		\begin{equation}\label{eq:snr_offset}
    			Z^{(Y)}_N \triangleq \left|h_{1,q} \right|^2 - \left|h_{2,q} \right|^2 {\mathbbm{1}_{\{N \neq 0\}}}.
    		\end{equation}  
    		The event $\{Z^{(Y)}_N<0\}$ indicates an SNR degradation at the OOB UE due to the IRS. In the following theorem, we show that, almost surely, $Z^{(Y)}_N$ is non-negative \textcolor{black}{by deriving the CCDF of $Z^{(Y)}_N$, $\bar{F}_{Z^{(Y)}_N}(z) \triangleq { \sf{Pr}}(Z^{(Y)}_N \geq z).$}
      		\begin{theorem}\label{thm:exact_ccdf}
    			The probability of outage at  UE-$q$ served by an (OOB) operator Y in the sub-6 GHz band, when an IRS is optimized to serve the UEs of operator X, is given by
    			\begin{equation}
    P^\rho_{q,\text{out}} = 1 - e^{-\dfrac{\rho}{N\beta_{r,q}+\beta_{d,q}}}.
        			\end{equation} 
    			Further, for reasonably large $N$,\footnote{We will later numerically show that the result holds even for $N\ge 4$.}  	
    			the CCDF $\bar{F}_{Z^{(Y)}_N}(z)$ of $Z^{(Y)}_N$is given by
    			\begin{equation}\label{eq:thm_exact_ccdf_sub_6GHz}
    			\!	\bar{F}_{Z^{(Y)}_N}(z)  = 
    			\begin{cases}
    					1-\dfrac{1}{N\tilde{\beta}+2}\times \! e^{\left(\dfrac{z}{\beta_{d,q}}\right)}, \hspace{0.1cm} & \!\!\!\! \mathrm{if}  \ z < 0,\\
    					\left(\dfrac{N\tilde{\beta}+1}{N\tilde{\beta}+2}\right)\times \! e^{-\left(\dfrac{z}{\beta_{d,q}\left(1+N\tilde{\beta}\right)}\right)},   & \!\! \!\!\mathrm{if} \  z \geq 0.\\
    				\end{cases}
    			\end{equation} where $\tilde{\beta} \triangleq
    		\beta_{r,q}/\beta_{d,q}$.
    		\end{theorem}
    		\begin{proof}
    			We can obtain $P^\rho_{q,\text{out}} $ using the CCDF of an exponential random variable. 
    			For~\eqref{eq:thm_exact_ccdf_sub_6GHz}, see Appendix~\ref{sec:ccdf_proof}.\!\!
    		\end{proof}
    		From Theorem~\ref{thm:exact_ccdf}, we see that the outage probability is strictly monotonically decreasing with $N$, which shows that the IRS only improves the OOB performance. In fact, from a first-order Taylor series approximation, we can show that $P^\rho_{q,\text{out}} \approx {\rho}/{(N\beta_{r,q}+\beta_{d,q})}$, i.e., the outage probability decreases linearly in $N$. Further, since 
    		$\bar{F}_{Z^{(Y)}_N}(0) = 1 - {1}/\left({2+N\tilde{\beta}}\right),$
    		 for a given $\tilde{\beta}$, 
    		the probability that the SNR/gain offset in~\eqref{eq:snr_offset} is negative  decays as $\mathcal{O}\left({1}/{N}\right)$. 
    		Also, $\bar{F}_{Z^{(Y)}_{N'}}(z) \geq \bar{F}_{Z^{(Y)}_{N''}}(z)$ for all $z$ and $N' \!> \!N''$. 
    		Consequently, we have the following proposition. 
    		\begin{proposition}\label{sec:prop_stochastic_dominance_sub_6_GHz}
    			For any $M, N \in \mathbb{N}$ with $M>N$, the random variable $Z^{(Y)}_{M}$ \emph{stochastically dominates}\footnote{A real-valued random variable $X$ is stochastically larger than, or dominates,  the real-valued random variable $Y$, written $X >_{st} Y$, if
    				\begin{equation}
    					{\sf {Pr}}(X>a) \geq {\sf {Pr}}(Y>a), \hspace{0.2cm} \text{for all }a.
    				\end{equation} Note that, if the random variables $X$ and $Y$ have CCDFs $\bar{F}$ and $\bar{G}$, respectively, then $X >_{st} Y \Longleftrightarrow \bar{F}(a) \geq \bar{G}(a)$ $\forall$ $a\in \mathbb{R}$~\cite{Ross_2014_ProbabiltyBook}.} $Z^{(Y)}_{N}$. In particular, the channel gain in the presence of the IRS stochastically dominates the channel gain in its absence.
    		\end{proposition}
    		The above proposition states that the random variables $\left\{Z^{(Y)}_{n}\right\}_{n\in \mathbb{N}}$ form a sequence of \emph{stochastically larger} random variables as a function of the number of IRS elements, where $\mathbb{N}$ is the set of natural numbers. Thus,  the SNR offset increases with the number of IRS elements even at an OOB UE, i.e., the IRS only enhances the channel quality at an OOB UE at any point in time, with high probability. Therefore, the performance of OOB operators \emph{does not degrade} even when the operator is entirely oblivious to the presence of the IRS. Note that this holds true for any operator in the area; hence, no operator will be at a disadvantage due to the presence of an IRS being controlled by only one operator. 
    		\begin{remark} \label{outage_prob_ccdf_in_band_UE}
    		The outage probability and the CCDF of the SNR offset as in Theorem~\ref{thm:exact_ccdf} at the in-band UE (for whom the IRS is optimized) decay to zero as $\mathcal{O}(e^{-N})$ (see Appendix~\ref{sec:in-band_outage_CCDF_scaling_sketch}.) This improvement in decay is the benefit of using an optimized IRS. 
    		\end{remark}
    		Having asserted that IRS can only benefit every OOB operator in the sub-6 GHz bands, we next move on to understand the effect of IRS on the OOB performance in the mmWave bands of communications, where the IRSs can significantly boost the in-band operators' performance in establishing link connectivity, improving coverage, etc.\!\!

    	\section{OOB Performance: mmWave bands}\label{sec:mmwave_LOS_RR}
    	In this section, we evaluate the in-band and OOB performance of an IRS-aided mmWave system under RR scheduling of the UEs. The wireless channels in the mmWave bands are typically spatially sparse and directional, with a few propagation paths, as against the rich-scattering  sub-6 GHz channels. Consequently, the analysis in the sub-6 GHz bands does not extend to mmWave bands, and it is necessary to evaluate the OOB performance in the mmWave bands independently. We consider two different scenarios inspired by~\cite{Wang_TWC_2022_IRS_BA}, namely, the LoS and (L+)NLoS scenarios, as explained in Sec.~\ref{sec:contributions}.
    	In both cases, we assume that the channel (see~\eqref{eq:ch_model_mmwave_single_IRS}) to the OOB UEs is a mixture of LoS and NLoS paths.
    	\subsection{IRS Optimized for LoS Scenarios}\label{sec:single_path_mmwave_RR}
    	The dominant LoS channel of the in-band UE is~\cite{Wang_TCOM_2023,Wang_TVT_2020} 
    	 \begin{align}
    	 h_k &= N\gamma^{(1)}_{1,X}\gamma^{(2)}_{1,k}\mathbf{a}_N^H(\psi_{1,k})\boldsymbol{\Theta}\mathbf{a}_N^*(\phi_{1,X}) + h_{d,k} \\
    	&\stackrel{(a)}{=} N  \left(\! \gamma^{(1)}_{1,X}\gamma^{(2)}_{1,k}\left(\mathbf{a}_N^H(\phi_{1,X})\odot\mathbf{a}_N^H(\psi_{1,k})\right)\right)\!\boldsymbol{\theta} +  h_{d,k} \label{eq:mmwave_jio_single_path_simplified_channel},
    	\end{align}
    	where $\boldsymbol{\theta}=\text{diag}(\boldsymbol{\Theta}) \in \mathbb{C}^N$, and $(a)$ is obtained using the properties of the Hadamard product. Since the Hadamard product of two array response vectors is also an array vector but aligned in a different direction~\cite{Wei_Yu_TWC_2023_Sparse_Ch_estim_IRS}, we simplify~\eqref{eq:mmwave_jio_single_path_simplified_channel} as
    	\begin{equation}\label{eq:effective_channel_mmwave_single_path_airtel}
    		h_k = N  \gamma_{X,k}\mathbf{\dot{a}}_N^H(\omega^1_{X,k})\boldsymbol{\theta} +  h_{d,k},
    	\end{equation} where \textcolor{black}{$\omega^1_{X,k} \triangleq \sin_{(p)}^{-1}\left(\sin(\phi_{1,X}) + \sin(\psi_{1,k})\right)$},\footnote{\textcolor{black}{We define $\sin_{(p)}^{-1}(x)$ so that $x$ is in the principal argument $[-1,1)$ as 				
    \begin{equation}
    			\!\!\!	\sin_{(p)}^{-1}(x) =\begin{cases}\!
    					\sin^{-1}(x - 2), & \mathrm{ if } \ x \geq 1,\\ \!
    					\sin^{-1}(x), & \mathrm{ if } \ x \in [-1,1),\\ \!
    					\sin^{-1}(x+2), & \mathrm{ if } \ x < -1.
    				\end{cases} 
    			\end{equation}}} and  $\gamma_{X,k} \triangleq \gamma^{(1)}_{1,X}\gamma^{(2)}_{1,k}$. Here, $\mathbf{\dot{a}}_N$ is an array vector that is normalized by $N$ instead of $\sqrt{N}$ (see~\eqref{eq:array_vector_template}), and so $\mathbf{\dot{a}}_N(\cdot) = \frac{1}{\sqrt{N}}\mathbf{a}_N(\cdot)$. Then by the Cauchy-Schwarz (CS) inequality, the $n$th entry of the  optimal IRS configuration vector $\boldsymbol{\theta}^{\mathrm{opt}}$ which maximizes the channel gain $|h_k|^2$ is $\theta_n = e^{j\left(\angle h_{d,k}-\pi(n-1)\omega^1_{X,k}-\angle\gamma_{X,k} \right)}$, and hence the optimal IRS vector is 
	  	\begin{equation}\label{eq:IRS_optimal_single_path_mmwave}
    			\boldsymbol{\theta}^{\mathrm{opt}} = \dfrac{h_{d,k}\gamma^*_{X,k}}{\left|h_{d,k}\gamma_{X,k} \right|} \times N \mathbf{\dot{a}}_N(\omega^1_{X,k}).
    		\end{equation}
    	From~\eqref{eq:IRS_optimal_single_path_mmwave}, it is clear that the IRS vector is \emph{directional}; it is aligned in the direction of the in-band UE's channel. To illustrate this, we plot the correlation function $\mathbb{E}\left[\left|\mathbf{\dot{a}}^H(\nu)\boldsymbol{\theta}^{\mathrm{opt}}\right|^2\right]$ versus $\nu \in \boldsymbol{\Phi}$ in Fig.~\ref{fig:correlation_single_path_IB}, where $\boldsymbol{\theta}$ is set to align with an arbitrary in-band UE with  $\omega^1_{X,k} = 0.52$. The expectation is with respect to the channel fading coefficients. 
    	    	Clearly, when $\nu = 0.52 =\omega^1_{X,k}$, the function takes its maximum value of $1$. 
    	This confirms that the IRS aligns with the in-band channel to which it is optimized. Contrariwise, the IRS phase shift as optimized by BS-X will be aligned at a \emph{random} direction from any OOB UE's view with the distribution given by $\mathcal{U}_{\mathcal{A}}$ in~\eqref{eq:codeboook_ditbn}. We now characterize the ergodic sum-SEs of both operators.
    		\begin{theorem}\label{thm:rate_characterization_mmwave_single_path_IB}
    		Under the Saleh-Valenzuela LoS model in the mmWave channels, with RR scheduling, and when the IRS is optimized to serve the UEs of operator X, the ergodic sum-SEs of  operators X and Y scale~as
    		\begin{multline}\label{eq:mmwave_rate_airtel_rr_single_path}
    			\bar{S}_1^{(X)} \approx \frac{1}{K}\sum_{k=1}^K \log_{2}\bigg(1 + \bigg[N^2\beta_{r,k}   \\  \left. \left. + \ N \left(\frac{\pi^{3/2}}{4}\sqrt{\beta_{d,k}\beta_{r,k}}\right)+ \beta_{d,k}\right]\frac{P}{\sigma^2} \right) ,
    			\end{multline} and
    		\begin{multline}\label{eq:mmwave_rate_jio_rr_single_path}
    			\bar{S}_1^{(Y)} \approx \frac{1}{Q}\sum_{q=1}^Q\left( \dfrac{\bar{L}}{N} \log_{2}\left(1 + \left[\dfrac{N^2}{\bar{L}} \beta_{r,q} + \beta_{d,q}\right]\frac{P}{\sigma^2} \right)  \right. \\ \left. + \left(1-\dfrac{\bar{L}}{N}\right)\log_2\left(1+\beta_{d,q}\frac{P}{\sigma^2}\right)\right),
    		\end{multline} respectively, where $\bar{L} \triangleq \min\left\{L,N\right\}$, and $L \triangleq L_1L_2$.
    	\end{theorem}
    	\begin{proof} 
    	We derive the SEs of both operators separately below.
    	\subsubsection{Ergodic sum-SE of operator X}
    	Let the average SE seen by UE-$k$ be  $\langle S_{k,1}^{(X)} \rangle$. Then the ergodic sum-SE of operator X is $\bar{S}_1^{(X)} = \frac{1}{K}\sum_{k=1}^K\langle S_{k,1}^{(X)} \rangle$.  By Jensen's inequality,
    	\begin{equation}\label{eq:jensen_mmwave_single_path}
    		\!\!\langle S_{k,1}^{(X)} \rangle  \leq \log_2\!\left(\!1+\mathbb{E}\left[|h_k|^2\right]\!\frac{P}{\sigma^2}\right).
    	\end{equation}  We recognize that
    	\begin{align}
    		\!\!\left|h_k\right|^2 &= \left|N  \gamma_{X,k}\mathbf{\dot{a}}_N^H(\omega^1_{X,k})\boldsymbol{\theta}^{\mathrm{opt}} +  h_{d,k}\right|^2\\ & \stackrel{(a)}{=} {\left| h_{d,k} \!+\! N^2\dfrac{h_{d,k}\left|\gamma_{X,k}\right|^2}{\left|h_{d,k}\gamma_{X,k} \right|}\!\times\! \mathbf{\dot{a}}_N^H(\omega^1_{X,k})\mathbf{\dot{a}}_N(\omega^1_{X,k})\right|}^2\!,\!\!
    	\end{align} 
	where $(a)$ is due to~\eqref{eq:IRS_optimal_single_path_mmwave}. Further, using $\mathbf{\dot{a}}_N^H(\omega^1_{X,k})\mathbf{\dot{a}}_N(\omega^1_{X,k}) = \frac{1}{N}$, we can compute  $\mathbb{E}\left[\left|h_k\right|^2\right]= \mathbb{E}\left[\left|\left|h_{d,k}\right| +N\left|\gamma_{X,k}\right|\right|^2\right]=\mathbb{E}\left[\left|h_{d,k}\right|^2\right] +N^2\mathbb{E}\left[\left|\gamma_{X,k}\right|^2\right] + 2N\mathbb{E}\left[\left|h_{d,k}\right|\left|\gamma_{X,k}\right| \right]$. Using the results from Sec.~\ref{sec:ch_model_mmwave}, this can be simplified as 
    	\begin{equation}
    		\mathbb{E}\left[\left|h_k\right|^2\right] = N^2 \beta_{r,k} + N\left(\dfrac{\pi^{3/2}}{4}\right)\sqrt{\beta_{r,k}\beta_{d,k}} + \beta_{d,k}.
    	\end{equation} Substituting the above in~\eqref{eq:jensen_mmwave_single_path}, and plugging it into the expression for $\bar{S}_1^{(X)}$ yields~\eqref{eq:mmwave_rate_airtel_rr_single_path}.	
    	
	\subsubsection{Ergodic sum-SE of operator Y}\label{app:sec:oob_ergoidc_rate-mmwave_single_path}	
    	Recall that the ergodic sum-SE of operator Y is $\bar{S}_1^{(Y)} = \frac{1}{Q}\sum_{q=1}^Q\langle S_{q,1}^{(Y)} \rangle$. Then, the channel seen by an arbitrary OOB UE, say UE-$q$, can be derived in a similar manner as~\eqref{eq:effective_channel_mmwave_single_path_airtel} and is given by
    	\begin{equation}
    		h_q =  h_{d,q}+\dfrac{N}{\sqrt{L}} \sum\nolimits_{l=1}^{L} \gamma^{(1)}_{l,Y}\gamma^{(2)}_{l,q}\mathbf{\dot{a}}_N^H(\omega^l_{Y,q})\boldsymbol{\theta}^{\mathrm{opt}},
    	\end{equation} where  $\boldsymbol{\theta}^{\mathrm{opt}}$ is given by~\eqref{eq:IRS_optimal_single_path_mmwave}, and $\gamma^{(1)}_{l,Y},\gamma^{(2)}_{l,q}$ denote the channel gains of the paths between BS-IRS, and IRS-UE, respectively, corresponding to the $l$th cascaded path. Then we have 
    	\begin{align}
    		h_q &= h_{d,q} + \dfrac{N^2}{\sqrt{L}}\sum\nolimits_{l=1}^{L} \!\!\!\gamma^{(1)}_{l,Y}\gamma^{(2)}_{l,q}\dfrac{h_{d,k}\gamma^*_{X,k}}{\left|h_{d,k}\gamma_{X,k} \right|}\mathbf{\dot{a}}_N^H(\omega^l_{Y,q})\mathbf{\dot{a}}_N(\omega^1_{X,k}) \nonumber \\ &=  h_{d,q} + \dfrac{N^2}{\sqrt{L}} \sum\nolimits_{l=1}^{L} \gamma^{(1)}_{l,Y}\gamma^{(2)}_{l,q}\frac{e^{j\angle h_{d,k}}}{e^{j\angle \gamma_{X,k}}}\mathbf{\dot{a}}_N^H(\omega^l_{Y,q})\mathbf{\dot{a}}_N(\omega^1_{X,k}) \nonumber\\
    		&\stackrel{d}{=}  h_{d,q} + \dfrac{N^2}{\sqrt{L}} \sum\nolimits_{l=1}^{L} \gamma^{(1)}_{l,Y}\gamma^{(2)}_{l,q}\mathbf{\dot{a}}_N^H(\omega^l_{Y,q})\mathbf{\dot{a}}_N(\omega^1_{X,k}), \label{eq:final_OOB_channel_mmwave_single_path}
    	\end{align} where, in the last step, we used the fact that \textcolor{black}{the product term $\gamma^{(1)}_{l,Y}\gamma^{(2)}_{l,q}$} is circularly symmetric. We also have
    	\begin{equation}\label{eq:jensen_mmwave_single_path_OOB}
    		\langle S_{q,1}^{(Y)} \rangle \triangleq \mathbb{E}\left[ S_{q,1}^{(Y)}\right] =  \mathbb{E}\left[\log_2\left(1+|h_q|^2\frac{P}{\sigma^2}\right)\right] .
		    	\end{equation} 
    	Consider $L< N$. Now, let $\mathcal{E}_1$ be the event that one of the $L$ angles of the OOB channel ``matches'' with the IRS angle, and $\mathcal{E}_0$ be the event that none of the $L$ angles of the OOB channel match with the IRS angle. Then, we can write~\eqref{eq:jensen_mmwave_single_path_OOB} as
    	\begin{equation}
    	\!\!\!\!	\langle S_{q,1}^{(Y)} \rangle
	    	\stackrel{(a)}{=} \sum_{i=0}^{1} \mathbb{E}\left[S_{q,1}^{(Y)}\mathbbm{1}_{\{\mathcal{E}_i\}}\right]  \stackrel{(b)}{=} \sum_{i=0}^{1} \mathbb{E}\left[S_{q,1}^{(Y)}\bigg\rvert\mathcal{E}_i\right] {\sf{Pr}}(\mathcal{E}_i), \!\!
    	\end{equation}
    	where $\mathbbm{1}_{\{\mathcal{E}_i\}}$ is the indicator for the occurrence of event $\mathcal{E}_i$, $(a)$ is because $\mathcal{E}_0$ and $\mathcal{E}_1$ are mutually exclusive and exhaustive; $(b)$ is due to the law of iterated expectations. Also, from~\eqref{eq:codeboook_ditbn}, and flat-top directivity of the IRS, ${\sf{Pr}}(\mathcal{E}_1) = L/N$, and ${\sf{Pr}}(\mathcal{E}_0) = \left(1-L/N\right)$. Then from~\eqref{eq:orthogonal_array_vectors}, and Jensen's inequality, we have
        	\begin{equation}
    		\!\!\!\!\mathbb{E}\left[S_{q,1}^{(Y)}\bigg\rvert\mathcal{E}_1\right] \!\leq  \log_2\!\left(\!\!1\!+\!\mathbb{E}\left[\left| h_{d,q} + \dfrac{N}{\sqrt{L}}  \gamma^{(1)}_{l^*,Y}\gamma^{(2)}_{l^*,q}\right|^2\right]\!\!\frac{P}{\sigma^2}\!\right)\!,\!\!\!
    	\end{equation}
    	\begin{equation} 
    		\vspace{-0.05cm}\text{and,    \hspace{0.75cm} }\mathbb{E}\left[S_{q,1}^{(Y)}\bigg\rvert\mathcal{E}_0\right] \leq\log_2\left(1+\mathbb{E}\left[\left|h_{d,q}\right|^2\right]\frac{P}{\sigma^2}\right),
    	\end{equation} where $l^* = \arg_l \left\{\mathbf{\dot{a}}_N^H(\omega^l_{Y,q})\mathbf{\dot{a}}_N(\omega^1_{X,k})=1/N\right\}$. We can show that (similar to  Sec.~\ref{proof:sub-6-jio-ergodi-rate}), $\mathbb{E}\left[|h_q|^2\bigg\rvert\mathcal{E}_1\right] = \left(\dfrac{N^2}{L} \beta_{r,q} + \beta_{d,q}\right)$, and $\mathbb{E}\left[|h_q|^2\bigg\rvert\mathcal{E}_0\right]=\beta_{d,q}$. Collecting all the terms and substituting in~\eqref{eq:jensen_mmwave_single_path_OOB} and using the expression for $\bar{S}_1^{(Y)}$ yields~\eqref{eq:mmwave_rate_jio_rr_single_path}. Now, when $L\geq N$, since an $N$-element IRS can steer at most $N$ resolvable beams, some of the paths will be clustered together, and the above analysis holds by replacing $L$ with $N$. We cover both cases by $\bar{L} = \min\{L,N\}$ in~\eqref{eq:mmwave_rate_jio_rr_single_path}. 
	    	\end{proof}
       From the above theorem, it is clear that an IRS can never degrade the OOB performance in LoS scenarios of the mmWave bands. Even when $L$ is small and $N$ is large, the second term in~\eqref{eq:mmwave_rate_jio_rr_single_path} remains, which is the achievable SE in the absence of IRS. Thus, almost surely, the IRS results in an OOB-SE that is at least as high as the SE seen in its absence.
    \begin{figure}
    		\vspace{-0.2cm}
    		\hspace{-0.2cm}
    		\begin{subfigure}{0.49\linewidth}
    	\includegraphics[width=1.09\linewidth]{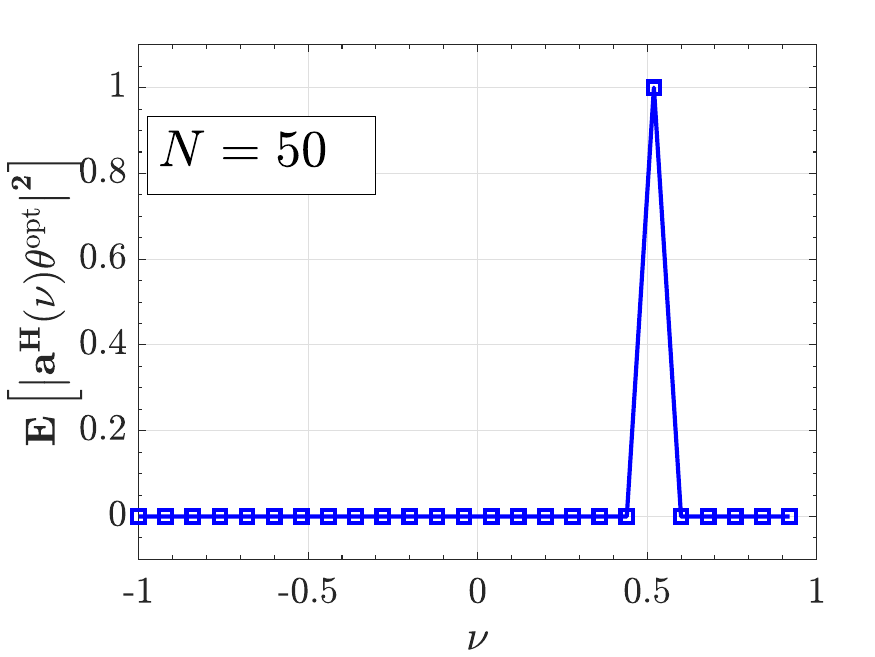}
    		\end{subfigure}
    		\begin{subfigure}{0.49\linewidth}
    		\includegraphics[width=1.09\linewidth]{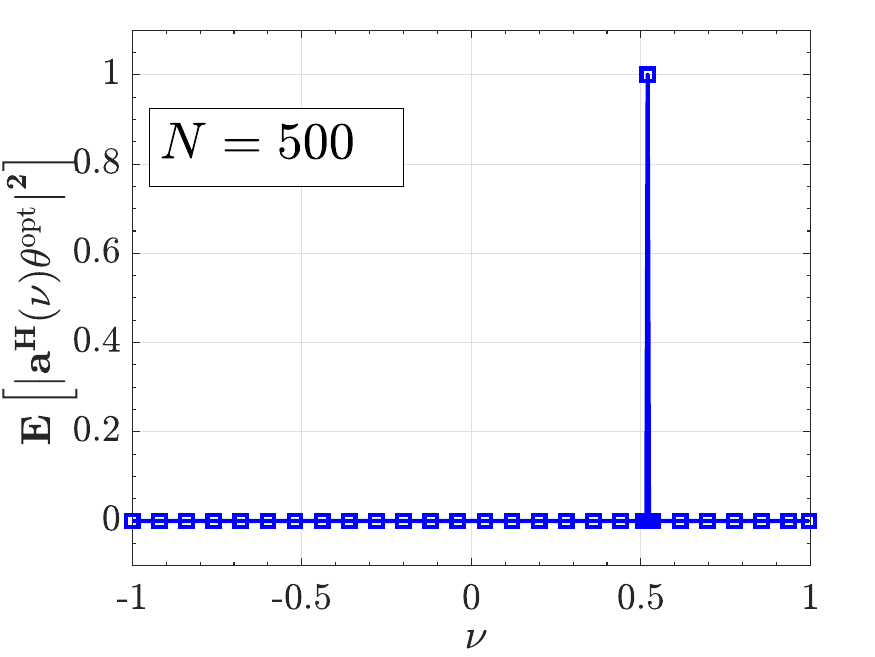}
    	\end{subfigure}
    	\caption{Correlation response of the IRS vector and array steering vectors pointing at different spatial angles, $\nu$, for (a) $N=50$ and (b) $N=500$. When $\nu=\omega^1_{X,k}$, the response attains its maximum value of $1$. }
	    	\label{fig:correlation_single_path_IB}
    	\vspace{-0.4cm}
    	\end{figure}
    	\begin{remark}[OOB performance as a function of $L$ for large $N$]\label{sec:favourable_unfavourable_region_mmwave_single_path}
    		From~\eqref{eq:mmwave_rate_jio_rr_single_path}, we can show that when $N \rightarrow \infty$ for a fixed and finite $L$, the SE is a unimodal function of $N$. Hence, there exists a $N=N_0$ which provides the maximum gain and distinguishes the \emph{favorable} region $(N\leq N_0)$ from the \emph{non-favorable} region $(N \geq N_0)$.\footnote{By non-favorable, we mean that the SE does not grow \textcolor{black}{linearly} with $N$ beyond $N=N_0$. However, the SE never drops below the achievable SE in the absence of the IRS.} Here, $N_0$ itself increases with $L$ because the probability of establishing connectivity via the IRS, i.e.,  ${\sf {Pr}}(\mathcal{E}_1)$ increases.\! 
    			\end{remark}
    	\begin{remark}[Distributed IRSs always benefit]
    	As noted in Remark~\ref{sec:favourable_unfavourable_region_mmwave_single_path}, the OOB UEs benefit more if the number of channel paths at these UEs is large. A natural way to obtain more paths is to deploy multiple distributed IRSs so that each IRS can provide an additional path to the UEs. This, in turn, increases the overall probability of benefitting the OOB UEs. We refer the reader to our follow-up work in~\cite{Yashvanth_distributedIRS_2024} for more details.\!\!
    	\end{remark}
    		    	\indent Similar to Theorem~\ref{thm:exact_ccdf}, the outage probability and CCDF of an OOB UE with/without IRS in mmWave bands is analyzed next.\!
    	\begin{theorem}\label{thm:ccdf_mmwave_OOB_single_path}
    		The probability of outage at  UE-$q$ served by an OOB operator Y in the mmWave bands, when an IRS is optimized to serve the UEs of operator X is
    		    		\begin{multline}\label{eq:ccdf_mmwave_single_path_jio}
    		\!\!\!	{F}_{|h_q|^2}(\rho) \triangleq {\sf {Pr}}(|h_q|^2<\rho) =1- e^{-\rho/\beta_{d,q}} \\ \!\!\! -   \dfrac{\bar{L}}{N}\left(\dfrac{\bar{L}e^{\frac{\bar{L}\beta_{d,q}}{N^2\beta_{r,q}} }}{N^2\beta_{r,q}} \mathcal{I}_0\left(\rho;\beta_{d,q},\dfrac{N^2}{\bar{L}}\beta_{r,q}\right)\!\!-  e^{-\rho/\beta_{d,q}}\!\right)\!, \!\!\!
    		\end{multline}
    	where $\bar{L} \triangleq \min\{L,N\}$,  and $
    	\ \mathcal{I}_0(x;c_1,c_2) \!\triangleq \!\stretchint{4.5ex}_{\!\!\!c_1}^{\infty}\!e^{\!-\left(\!\dfrac{x}{t}+\dfrac{t}{c_2}\right)}\!\text{d}t.$ 
    	Further,  define the random variables $G_1 = |h_q|^2$, and $G_0 = |h_{d,q}|^2$ (the channel gain in the \emph{presence} and \emph{absence} of the IRS, respectively.) Then, 
    	\begin{equation}\label{eq:stochastic_dominance_mmwave_single_path}
    	{\sf {Pr}}(G_1>\rho) \geq {\sf {Pr}}(G_0>\rho) \hspace{0.3cm} \forall \rho \in \mathbb{R}^+,
    	\end{equation} i.e., the channel gain of the OOB-UE in the presence of an IRS \emph{stochastically dominates} the channel gain in its absence.
	    	\end{theorem}
    	\begin{proof}
    		See Appendix~\ref{app:proof_ccdf-single_path_mmwave}.	
    	\end{proof}
    The above theorem characterizes the instantaneous behavior of the OOB UEs' channels in the LoS scenarios and shows that the IRS never degrades the performance of OOB UEs.
    	\subsection{IRS Optimized for (L+)NLoS Scenarios}\label{sec:multiple_paths_mmwave_RR}
    	In the previous section, we studied the OOB performance when the IRS is only optimized to the dominant path in the cascaded channels of the scheduled in-band UEs and concluded that the OOB performance does not degrade due to the IRS. However, when the IRS is programmed to align jointly along all the paths of the in-band UE's channel, it remains unclear whether an IRS can still benefit an OOB operator due to the unit-modulus preserving property of IRS coefficients. In fact, the directional response of the IRS is not easy to characterize since the IRS does not align completely along a single-channel direction. To that end, we first determine the directional response of the IRS when it is optimized to a channel comprising multiple paths. We can write the  channel 
    	to the in-band UE-$k$ \textcolor{black}{with $L$ resolvable spatial paths} as in the following equation.\footnote{\textcolor{black}{For the sake of exposition, we consider that all cascaded paths through the IRS have the same average energy, similar to~\cite{Bai_Lin_Access_2017}. However, our results can be directly extended to scenarios where the paths have unequal energies.}}
    	\begin{equation}\label{eq:effective_channel_mmwave_multiple_path_airtel}
    		h_k =  h_{d,k} + \frac{N}{\sqrt{L}}\sum\nolimits _{l=1}^{L} \gamma^{(1)}_{l,X}\gamma^{(2)}_{l,k}\mathbf{\dot{a}}_N^H(\omega^l_{X,k})\boldsymbol{\theta}.
    	\end{equation}
    	Then, we can show that $n$th element of the optimal IRS configuration $\boldsymbol{\theta}^{\mathrm{opt}}$ is given by (using CS inequality)
    	\begin{equation}
    	\theta^{\mathrm{opt}}_n = e^{j\angle h_{d,k}}\dfrac{\sum_{l=1}^{L} \gamma^{(1)*}_{l,X}\gamma^{(2)*}_{l,k}e^{-j(n-1)\pi\omega^l_{X,k}}}{\left|\sum_{l=1}^{L} \gamma^{(1)*}_{l,X}\gamma^{(2)*}_{l,k}e^{-j(n-1)\pi\omega^l_{X,k}}\right|}.
    	\end{equation}
    	Hence, the optimal IRS vector $\boldsymbol{\theta}^{\mathrm{opt}} $ is 
	    	\begin{equation}\label{eq:optimal_IRS_mmwave_multi_path}
    	\dfrac{h_{d,k}}{|h_{d,k}|} \left(\sum_{l=1}^{L} \gamma^{(1)*}_{l,X}\gamma^{(2)*}_{l,k}\mathbf{\dot{a}}_N	(\omega^l_{X,k})\right) \odot\dfrac{1}{\left|\sum\limits_{l=1}^{L} \gamma^{(1)}_{l,X}\gamma^{(2)}_{l,k}\mathbf{\dot{a}}_N(\omega^l_{X,k})\right|},
    	\end{equation}
    	where $|\mathbf{x}|$ is the vector containing the magnitudes of the entries of $\mathbf{x}$ and $1/|\mathbf{x}|$ is an entry-wise inverse.\!\! \\
    	\indent \emph{Directional Response of the IRS:} 
    	To determine the directional response of an IRS optimized for (L+)NLoS scenarios, similar to Sec.~\ref{sec:single_path_mmwave_RR}, we evaluate the correlation function: $\rho_{\nu,\theta} \triangleq N\mathbb{E}\left[\left|\mathbf{\dot{a}}^H(\nu) \boldsymbol{\theta}^{\mathrm{opt}}\right|\right]$ at various angles of $\nu$. We then focus on the values of $\nu$ that match the channel angles of UE-$k$ and ascertain the distribution of the IRS-reflected energy along these directions. We have the following lemma.\!\!
    	\begin{lemma}\label{lemma_correlation_function_IRS_mmwave_multiple_path} \label{lemma:correlation_function}
    	The optimal IRS configuration as in~\eqref{eq:optimal_IRS_mmwave_multi_path} has the following spatial amplitude response $\rho_{\nu,\theta}$ as defined above:
    	\begin{equation}
    	\hspace{-0.12cm}
    	\rho_{\nu,\theta} =
    	\begin{cases}
    	 \Omega\left(\dfrac{N}{\sqrt{L}}\right) + o(N), \hfill \mathrm{ if } \ \nu \in \left\{\omega^1_{X,k},\!\! \ldots,\omega^L_{X,k} \right\},\hspace{-0.2cm} \\
    	 o(N), \hfill \mathrm{ if } \ \nu \in \mathbf{\Phi} \setminus \left\{\omega^1_{X,k}, \ldots,\omega^L_{X,k} \right\}.\hspace{-0.2cm}
    	\end{cases}\hspace{-0.2cm}
    	\end{equation}
    	\end{lemma}
    	\begin{proof}
    		See Appendix~\ref{app_proof_correlation_function}.
    	\end{proof}
    	The above Lemma shows that when the IRS is optimized to align jointly along $L$ directions, then it has a spatial energy response of $1/L$ in each of these $L$ directions and negligible response in other directions. We perform a similar experiment as in the LoS scenario to numerically verify Lemma~\ref{lemma_correlation_function_IRS_mmwave_multiple_path}.  In Fig.~\ref{fig:correlation_multiple_path_IB}, we plot the normalized correlation response $\rho_{\nu,\theta}/N$ as function of $\nu \in \boldsymbol{\Phi}$, for $L=2$ and $L=3$. We consider that the IRS is optimized to an in-band UE whose equal-gain paths have channel angles drawn from $\mathcal{L}_2 = \{-0.23,0.54\}$, and $\mathcal{L}_3 = \{-0.23, 0.06,0.54\}$, for $L=2$, and $L=3$, respectively. We see that whenever $\nu \in \mathcal{L}_2$ for $L=2$ ($\mathcal{L}_3$ for $L=3$), the response is maximum, and is nearly $1/\sqrt{L}$, and is $0$, otherwise.
    	\begin{figure}
    	\vspace{-0.3cm}
    	\hspace{-0.3cm}
    		\begin{subfigure}{0.49\linewidth}
    			\includegraphics[width=1.125\linewidth]{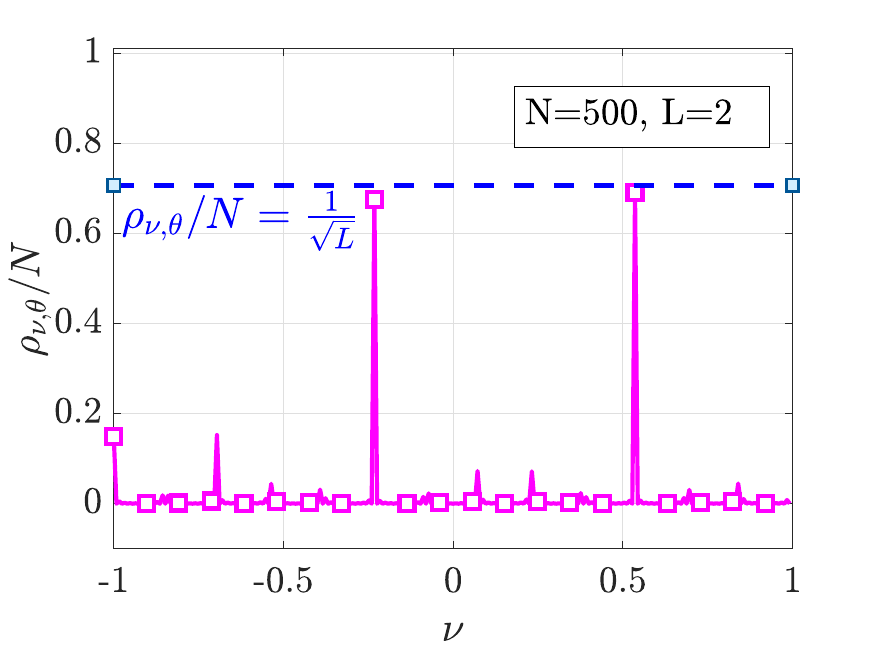}
    			\caption{$L=2.$}
    		\end{subfigure}
    		\begin{subfigure}{0.49\linewidth}
    			\includegraphics[width=1.125\linewidth]{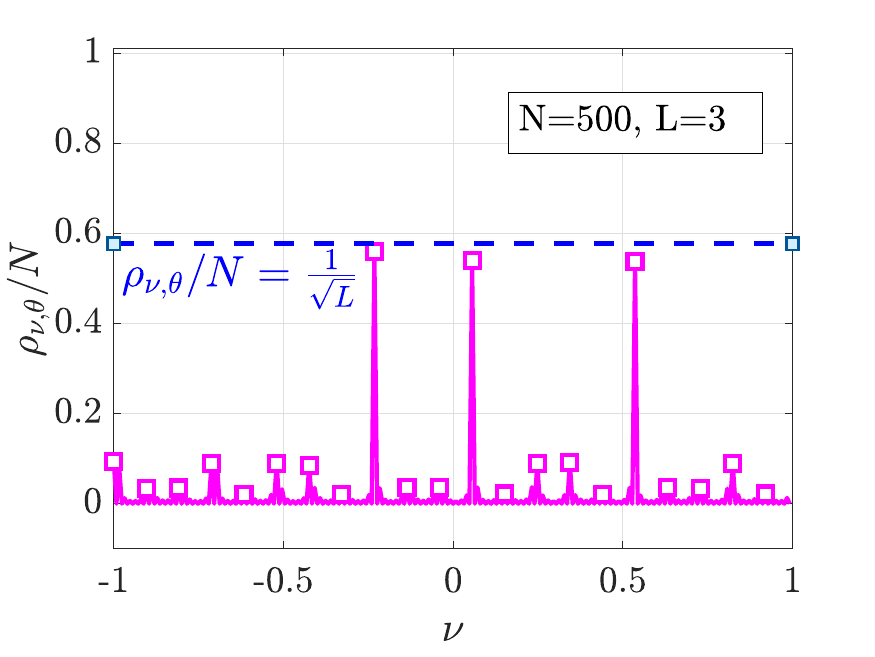}
    			\caption{$L=3.$}
    		\end{subfigure}
    		\caption{Normalized Correlation response of the IRS  (in (L+)NLoS scenarios) and array  vectors at different $\nu$ for $N=500$ with (a) $L=2$, and (b) $L=3$. When $\nu \in \mathcal{L}_2\text{ (or) }\mathcal{L}_3$, the response, $\rho_{\nu,\theta}/N$ peaks and $\approx 1/\sqrt{L}$.}    		\label{fig:correlation_multiple_path_IB}
    		\vspace{-0.4cm}
    	\end{figure}
    
    	We now provide the ergodic sum-SEs of the system when the IRS is optimized for (L+)NLoS scenarios.\!\!
    	\begin{theorem}\label{thm:ergodic_SE_mmwave_multiple_paths}
    	Under the Saleh-Valenzuela (L+)NLoS model in the mmWave channels, with RR scheduling, and when the IRS is optimized to serve the UEs of operator X, the ergodic sum-SEs of  operators X and Y scale~as 
    	    	\begin{equation}\label{eq:ergodic_sum_SE_mmwave_airtel_multiple_paths}
    	\!\!\bar{S}_L^{(X)} \!\approx\! \frac{1}{K}\!\sum_{k=1}^K \log_2\!\left(\!1\!+\!\left(N^2\beta_{r,k}\!+\!N\!\sqrt{\pi\beta_{d,k}\beta_{r,k}}\!+\beta_{d,k}\right)\!\dfrac{P}{\sigma^2}\right)\!,\!\!
    	\end{equation} and as~\eqref{eq:ergodic_sum_SE_mmwave_jio_multiple_paths} on the top of next page, respectively, where $i_0 \triangleq \max\{0,2L-N\}$.
        	\end{theorem}
    	\begin{figure*}[t]
    	\vspace{-0.5cm}
    	\begin{align}\label{eq:ergodic_sum_SE_mmwave_jio_multiple_paths}
    		\!\!\!\!\bar{S}_L^{(Y)} \approx \begin{cases}
    			\dfrac{1}{Q}\sum\limits_{q=1}^{Q}\sum\limits_{i=i_0}^L\dfrac{\mybinom[0.9]{L}{i}\mybinom[0.9]{N-L}{L-i}}{\mybinom[0.9]{N}{L}} \times  \log_2\left(1 + \left(\beta_{d,q} + i\dfrac{N^2}{L^2}\beta_{r,q}\right)\dfrac{P}{\sigma^2}\right), & \mathrm{ if } \ \  L < N, \\
    			\dfrac{1}{Q}\sum\limits_{q=1}^{Q}\log_2\left(1+\left(\beta_{d,q} + N\beta_{r,q}\right)\dfrac{P}{\sigma^2}\right),& \mathrm{ if } \ \  L \geq N.
    		\end{cases}
    	\end{align}
    	\rule{\textwidth}{0.3mm}
    	\vspace{-0.6cm}
    \end{figure*}
    	\vspace{-0.3cm}
    	\begin{proof}
    	We derive the SEs of both operators separately below.
    	\subsubsection{Ergodic sum-SE of operator X}
    	As  before, we first compute the ergodic SE of UE-$k$, $\langle S^{(X)}_{k,L} \rangle$, and then evaluate the ergodic sum-SE of operator X. Using~\eqref{eq:effective_channel_mmwave_multiple_path_airtel}~and~\eqref{eq:optimal_IRS_mmwave_multi_path}, the ergodic SE of UE-$k$ is bounded by
        		$\langle S^{(X)}_{k,L} \rangle \leq \!\log_2\!\left(\!1\!+\!\mathbb{E}\!\left[\left||h_{d,k}|+  \frac{1}{\sqrt{L}}\!\sum\limits_{n=1}^N\!\left|\sum\limits_{l=1}^L\!\!\gamma^{(1)}_{l,X}\gamma^{(2)}_{l,K}e^{j\pi(n-1)\omega^l_{X,k}}\right|\right|^2\right]\!\dfrac{P}{\sigma^2}\right)\!.\!\!$
    	Using the fact that the distribution of $\gamma^{(1)}_{l,X}\gamma^{(2)}_{l,K}$ is circularly symmetric, we can rewrite the inner term as
    	$ \mathbb{E}\left[\left||h_{d,k}|+\frac{N}{\sqrt{L}}\left|\sum\limits_{l=1}^L\!\gamma^{(1)}_{l,X}\gamma^{(2)}_{l,K}\right|\right|^2\right]= \mathbb{E}\left[|h_{d,k}|^2+\dfrac{N^2}{L}\left|\sum\limits_{l=1}^L\gamma^{(1)}_{l,X}\gamma^{(2)}_{l,K}\right|^2  +  \dfrac{2N}{\sqrt{L}}|h_{d,k}|\left|\sum\limits_{l=1}^L \gamma^{(1)}_{l,X}\gamma^{(2)}_{l,K}\right|\!\right].$
    	\underline{\emph{Term I}}: It is straightforward to see that $\mathbb{E}[|h_{d,k}|^2] = \beta_{d,k}$.\\
    	\underline{\emph{Term II}}: We can write
    	$	\!\!\mathbb{E}\!\left[	\left|\sum\nolimits_{l=1}^L\!\gamma^{(1)}_{l,X}\gamma^{(2)}_{l,K}\right|^2\right] = \mathbb{E}\left[\mathop{\sum\limits_{l=1}^L}\left|\gamma^{(1)}_{l,X}\gamma^{(2)}_{l,K}\right|^2 \!\right]+  \mathbb{E}\left[\mathop{\sum\limits_{l,\ell=1; { l\neq \ell}}^L} \gamma^{(1)}_{l,X}\gamma^{(2)}_{l,K}\gamma^{(1)*}_{\ell,X}\gamma^{(2)*}_{\ell,K}\right] \stackrel{(a)}{=} \!L\beta_{r,k},\!$
    	where $(a)$ is due to the zero-mean and independence of $\left\{\gamma^{(1)}_{l,X},\gamma^{(2)}_{l,K}\right\}_l$ across the spatial paths.\\
    	\underline{\emph{Term III}}: We have the following sequence of relations
    	$	\mathbb{E}\left\{|h_{d,k}|\left|\sum\limits_{l=1}^L\gamma^{(1)}_{l,X}\gamma^{(2)}_{l,K}\right|\right\} \stackrel{(a)}{\leq}  \mathbb{E}|h_{d,k}| \sqrt{\mathbb{E}\left[\left|\sum\limits_{l=1}^L\gamma^{(1)}_{l,X}\gamma^{(2)}_{l,K}\right|^2\right]} \nonumber  = \dfrac{1}{2}\sqrt{\pi \beta_{d,k}}\cdot\sqrt{L\beta_{r,k}},$
    	where $(a)$ follows by using the independence of $h_{d,k}$ and $\sum_{l=1}^L\gamma^{(1)}_{l,X}\gamma^{(2)}_{l,K}$, followed by the Jensen's inequality. Collecting all the terms  yields~\eqref{eq:ergodic_sum_SE_mmwave_airtel_multiple_paths}.\!\!\!
    	\subsubsection{Ergodic sum-SE of operator Y} Consider the case $L < N$.
    	As above, the channel seen at the UE-$q$ is 
    	\begin{equation}\label{eq:effective_channel_mmwave_multiple_path_jio}
    		h_q =  h_{d,q} + \frac{N}{\sqrt{L}}\sum\nolimits_{l=1}^{L} \gamma^{(1)}_{l,Y}\gamma^{(2)}_{l,q}\mathbf{\dot{a}}_N^H(\omega^l_{Y,q})\boldsymbol{\theta},
    	\end{equation}
    	where $\boldsymbol{\theta}$ is optimized by operator-X and is given by~\eqref{eq:optimal_IRS_mmwave_multi_path}. The IRS vector aligns along $L$ directions of  $\boldsymbol{\Phi}$, while the OOB UE's channel is oriented along $L$  directions independent of the IRS. We let the random variable 
    	$X$ denote the instantaneous SE at UE-$q$, and $Y$ denotes the number of matching (spatial) paths between the UE-$q$'s channel and the IRS response. Clearly, the possible support of $Y$ is the set of integers from $0$ to $L$. We now compute $\langle S^{(Y)}_{q,L} \rangle = \mathbb{E}[X] $ as 
    	\begin{equation}\label{eq:ergodic_SE_mmwave_multiple_path_template}
    		\!\mathbb{E}[X] = \mathbb{E}_Y[\mathbb{E}_X[X|Y]] = \sum\nolimits_i \mathbb{E}_X[X|Y=i]{\sf {Pr}}(Y=i).\!
    	\end{equation} 
	    	When $L < N$, we can show that the probability mass function (pmf) of $Y$ is given by
    	\begin{equation}\label{eq:prob_oob_NLOS}
    		{\sf{Pr}}(Y=i) = \dfrac{\binom{L}{i}\binom{N-L}{L-i}}{\binom{N}{L}}, \hspace{0.2cm} i = i_0, i_0+1,\ldots,L,
    	\end{equation} and $0$ otherwise, where \textcolor{black}{$\binom{x}{y}$ is the usual binomial coefficient, and }$i_0 = \max\{0,2L-N\}$. This is because, if $N-L\leq L$, not more than $N-L$ spatial paths of the OOB channel can be misaligned with the IRS, and hence, the support of the above pmf begins at $2L-N$.  Next, we compute the term $\mathbb{E}_X[X|Y=i]$ which is the average SE seen by UE-$q$ when exactly $i$ of the spatial paths are common between the IRS response and scheduled  UE-$q$'s channel angles. Then, $\mathbb{E}_X[X|Y=i]$ 
    	\begin{align}
    		\!\! &\stackrel{(a)}{\leq} \log_2\left(\! 1\! +\mathbb{E}\left[\left|h_{d,q}\! +\frac{N}{\sqrt{L}}\sum\nolimits_{{l'} \in \mathcal{I}_i}\! \gamma^{{l'}}_{Y,q}\mathbf{\dot{a}}_N^H(\omega^{l'}_{Y,q})\boldsymbol{\theta}\right|^{2}\right]\! \frac{P}{\sigma^2}\right) \! \! \nonumber \\ &\stackrel{(b)}{=} \log_2\left(1+\!\mathbb{E}\left[\left|h_{d,q}+\frac{N}{L}\sum\nolimits_{{l'} \in
    			\mathcal{I}_i}\gamma^{{l'}}_{Y,q}\right|^{2}\right] \frac{P}{\sigma^2}\!\right),\\[-0.3in] \nonumber
    		\end{align} 
    where $(a)$ is due to Jensen's inequality, $\gamma^{{l'}}_{Y,q} \triangleq \gamma^{(1)}_{{l'},Y}\gamma^{(2)}_{{l'},q}$, and $\mathcal{I}_i$ denotes the index set of the common path indices between the IRS and UE-$q$'s channel such that $|\mathcal{I}_i|=i$, and $(b)$ is due to Lemma~\ref{lemma:correlation_function}. By expanding the square term and using the statistics of the random variables, we can show that the above expectation becomes $\beta_{d,q}+i\frac{N^2}{L^2}\beta_{r,q}$. Plugging this into~\eqref{eq:ergodic_SE_mmwave_multiple_path_template} and evaluating the ergodic sum-SE yields~\eqref{eq:ergodic_sum_SE_mmwave_jio_multiple_paths} when $L<N$. Now, for $L \geq N$ (and hence $\bar{L}\triangleq \min\{L, N\} = N$), since the IRS can orient to at most $N$ beams, with probability $1$, every channel path of the UE is aligned to one of directions to which the IRS is steered. Thus, we have $\langle S^{(Y)}_{q,L} \rangle$
    	\begin{align}
    		\!\! &\leq \log_2\left(1+\mathbb{E}\left[\left|h_{d,q}+{\sqrt{N}}\sum\nolimits_{{l} =1}^{N}\gamma^{{l}}_{Y,q} \mathbf{\dot{a}}_N^H(\omega^{l}_{Y,q})\boldsymbol{\theta}\right|^2\right] \dfrac{P}{\sigma^2}\right) \nonumber \\ &= \log_2\left(1+\mathbb{E}\left[\left|h_{d,q}+\sum\nolimits_{{l}=1}^{N}\gamma^{{l}}_{Y,q}\right|^2\right] \vspace{-0.2cm}\dfrac{P}{\sigma^2}\right)\\&= 
    		\log_2\left(1+\left(\beta_{d,q} + N\beta_{r,q}\right) \frac{P}{\sigma^2}\right). \\[-0.25in] \nonumber
    	\end{align} The proof of~\eqref{eq:ergodic_sum_SE_mmwave_jio_multiple_paths} for $L \geq N$ now easily follows. 
    	\end{proof}
    \noindent Theorem~\ref{thm:ergodic_SE_mmwave_multiple_paths} shows that the IRS provides an $\mathcal{O}(\log_2(N))$ scaling of the SE at OOB UEs when $L \geq N$. In fact, in Sec.~\ref{sec:numerical_sections}, we numerically show that this bound can be improved to $L \geq \sqrt{N}$ while preserving the $\mathcal{O}(\log_2(N))$ growth of SE. In that case, for $L< \sqrt{N}$, the OOB-SE is shown to be log-sub-linear in $N$, which is still better than the SE in the absence of IRS.\footnote{\textcolor{black}{We reiterate that we need $L \geq \sqrt{N}$ only for $\mathcal{O}(\log_2(N))$ scaling of SE. More generally, the OOB-SE in mmWave bands scales as $\mathcal{O}(\log_2(N^\delta)), \delta\!>\!0$. Thus, an IRS strictly benefits the OOB system under all conditions.}} Subsequently, we show that the OOB sum-SE in (L+)NLoS scenarios is at least as good as in LoS scenarios and present further insights on the OOB performance in the presence of an uncontrolled IRS.\!\!\!
    
\section{Enhancement of OOB Performance Using Opportunistic Scheduling}\label{sec:PF_MR_schedulers}
    In the previous sections, we saw that an IRS that is optimally configured to serve the UEs of operator X also benefits operator Y, but to a lesser extent than the benefit to the UEs of operator X. In this section, we show that the benefit of the uncontrolled IRS to the UEs of operatory Y can be further enhanced using opportunistic user selection. Specifically, since the IRS is \emph{randomly} configured from an OOB UE's view, if there are sufficiently many OOB UEs, at least one of the UEs will experience an SNR that is close to the SNR when the IRS is optimized for that UE (i.e., when the IRS is in \emph{beamforming configuration} for that UE.) Then, opportunistically scheduling the UE for which the IRS is in beamforming configuration in every slot extracts multi-user diversity in the system and enhances the OOB performance better than RR-based UE scheduling~\cite{Viswanath_TIT_2002}. We now analyze the OOB performance for two such opportunistic schedulers: the proportional-fair (PF) and max-rate (MR) schedulers.\!\!
    \subsection{Multi-user Diversity for Operator Y using PF Scheduler}
    The PF scheduler serves UE-$q^*$ at time $t$, where~\cite{Viswanath_TIT_2002}
    \begin{equation}
    \!\!\!\!q^*(t) = \argmax_{q \in \{1,2,\ldots,Q\}} \frac{\log_2\left(1+|h_q(t)|^2P/\sigma^2\right)}{T_q(t)} \triangleq \frac{R_q(t)}{T_q(t)},
    \end{equation} 
    where $T_q(t)$ is the exponential moving average SE seen by UE-$q$ till time $t$ which is updated as 
    \begin{equation}
    T_q(t+1) = \begin{cases} \left(1-\frac{1}{\tau}\right) T_q(t) + \frac{1}{\tau} R_q(t), & \mathrm{if} \ \ q = q^*(t), \\ 
    \left(1-\frac{1}{\tau}\right) T_q(t), & \mathrm{if} \ \ q \neq q^*(t).
    \end{cases}
    \end{equation} Here, the parameter $\tau$ controls the trade-off between fairness and throughput~\cite[Sec.II]{Yashvanth_TSP_2023}. The following Lemma shows the achievable sum-SE of operator Y under PF scheduling.
    
    \begin{lemma}\label{lem:PF_sub_6}(\!\!\!\cite{Yashvanth_TSP_2023,Nadeem_WCL_2021})
    Under independent Rayleigh fading channels in the sub-$6$ GHz bands, when the IRS is optimized to serve the UEs of operator X, and operator Y uses the PF scheduler, the sum-SE of operator Y, $R^{(Y)}_{\mathrm{PF}}$, obeys
    \begin{equation}\label{eq_PF_oob_UE_SE}
    R^{(Y)}_{\mathrm{PF}} - \frac{1}{Q} \sum_{q=1}^{Q} \log_2\Bigg(\!\!1+ \left|\sum_{n=1}^N |f_n^Yg_{q,n}| + |h_{d,q}|\right|^2\!\!\!\frac{P}{\sigma^2}\!\!\Bigg) \longrightarrow 0,
    \end{equation}
    as $Q,\tau \longrightarrow \infty$. 
    \end{lemma}
    From~\eqref{eq_PF_oob_UE_SE} and~\eqref{eq_BF_rate_UE_k}, we see that UE-$q$ achieves the optimal SNR from the IRS even though the IRS is not explicitly programmed to be in beamforming configuration to this UE. Thus,  operator $Y$ also experiences $\mathcal{O}(N^2)$ gain in the SNR, free of cost, provided the number of UEs it serves is very large. However, one drawback of using the PF scheduler with large $\tau$ is that the latency in UE scheduling also becomes large~\cite{Viswanath_TIT_2002}. Thus, an OOB operator has to judiciously choose $\tau$ to balance the achievable sum rate with the latency in UE scheduling.

    \subsection{Multi-user Diversity for Operator Y using MR Scheduler}
    
    The MR scheduler serves UE-$q^*$ at time $t$, where 
    \begin{equation}\label{eq:MR_UE_pick}
    q^*(t) = \argmax_{q \in \{1,2,\ldots,Q\}} \log_2\left(1+\left|h_q(t)\right|^2P/\sigma^2\right).
    \end{equation} 
    Unlike the PF scheduler, an MR scheduler maximizes the overall system throughput, disregarding fairness across the UEs. As a result, the achievable benefit from the IRS when using an MR scheduler has a different flavor than the PF scheduler. Specifically, the random IRS configurations are used to enhance the dynamic range of fluctuations in the SNR at OOB UEs. Then, the ergodic sum-SE of the operator Y is 
    \begin{equation}\label{eq_MR_rate_template}
    \bar{R}^{(Y)}_{\mathrm{MR}} = \mathbb{E}\left[\log_2\left(1+\max_{q\in\{1,\ldots,Q\}}|h_q|^2P/\sigma^2\right)\right].
    \end{equation}
    A closed-form characterization of~\eqref{eq_MR_rate_template} is given below under a special case where $\{h_{q}\}_q$ form a set of i.i.d. random variables.
    \begin{lemma}\label{lem:MR_scheduler_ergodic_SE}(Theorem~$1$,\cite{Yashvanth_ICASSP_2023})
    Under i.i.d.\ Rayleigh fading channels in the sub-$6$ GHz bands, when the IRS is optimized to serve the UEs of operator X, and operator Y uses the MR scheduler, the ergodic sum-SE of operator Y, $\bar{R}^{(Y)}_{\mathrm{MR}}$, scales as
    \begin{equation}
    \bar{R}^{(Y)}_{\mathrm{MR}} \xrightarrow{Q \rightarrow \infty} \log_2\left(1+\log_e(Q)\left(N\beta_{r}+\beta_{d}\right)\frac{P}{\sigma^2}\right),
    \end{equation} where $\beta_r$, $\beta_d$ are the common direct \& cascaded path losses.
    \end{lemma}
    Thus, the MR scheduler leverages multi-user diversity to obtain a $\log_e(Q)$-fold improvement in SNR and focuses less on achieving the beamforming SE at the selected OOB UE.\\
    
    We note that practical aspects, such as the number of UEs required to achieve a target multi-user diversity, efficient feedback schemes to identify the best UE for scheduling, etc., are not discussed here due to space constraints. We refer the reader to~\cite[Sec. II]{Yashvanth_ICASSP_2023} and~\cite[Sec. III-A]{Yashvanth_TSP_2023} for more details.\!\\
    
    The results in Lemmas~\ref{lem:PF_sub_6}, \ref{lem:MR_scheduler_ergodic_SE} pertain to the sub-$6$ GHz bands. Similar results can also be derived for operator Y in mmWave bands using PF and MR schedulers, e.g., using the opportunistic scheme presented in~\cite[Sec.III-C]{Yashvanth_TSP_2023} and~\cite[Sec. III]{Yashvanth_ICASSP_2023}. 

    \section{Numerical Results and Discussion}\label{sec:numerical_sections}

    		In this section, we demonstrate all our results via Monte Carlo simulations. Let BS-X and BS-Y be located at $(0,50)$, and $(50,0)$ (in meters), the IRS is at $(1025,1025)$, and the UEs are randomly and uniformly located in a rectangular region with diagonally opposite corners $(950,950)$ and  $(1100,1100)$. The path loss in each link is modeled as $\beta = C_0\left(d_0/d\right)^\alpha$, where $C_0$ is the path loss at the reference distance $d_0$, $d$ is the distance of the link, and $\alpha$ is the path loss exponent. We let $d_0=1$ meter, and $C_0\!=\!-30$ dB and $\!-60$ dB for sub-6 GHz and mmWave bands, respectively. We use $\alpha=2,2$, and $4.5$ in the BS X/Y-IRS, IRS-UE, and BS X/Y-UE (direct) links, respectively. Finally, we consider $\!K\!=\!Q\!=\!10$ UEs served over $5000$ time slots.\!\! 
		
		\subsection{OOB Performance in sub-6 GHz Bands using RR Scheduler}

    		In Fig.~\ref{fig:SE_txt_SNR}, we plot the empirical ergodic sum-SE vs. the transmit SNR $\left(\gamma \triangleq P/\sigma^2\right)$ for both the operators as a function of $N$, the number of IRS elements, and in the sub-6 GHz frequency band.\footnote{We set the range of $\gamma$ in $110 - 160$ dB \textcolor{black}{for sub-6 GHz bands}. For e.g., when $N = 8$ and $\gamma = 160$ dB, the \emph{received} SNR is $\approx 5, -10$ dB
    		with and without an IRS, respectively, at an OOB UE at $(1000, 1000)$.}  We also plot the sum-SE obtained from the analytical expressions in Theorem~\ref{thm:rate_characterization}. We see that the IRS enhances the sum-SEs of both operators, although operator X benefits more from the IRS. We also see that the improvement in SE with $N$ is log-quadratic for operator X, while it is log-linear for operator Y, as expected. For example, the gap between the $N=64$ and $N=256$ curves for operator X is about $4$ bits ($\approx \log_2(4^2) = 4$), the gap is about $2$ bits ($\approx \log_2(4)$) for operator Y. Also, the analytical expressions tightly match the simulated values, i.e., the approximation error introduced using Jensen's inequality is very small.\\
    		\indent Next, in Fig.~\ref{fig:SE_logN}, we examine the effect of the number of IRS elements, $N$, by plotting the ergodic sum-SE vs. $\log_2(N)$ for transmit SNRs of $130$ dB and $150$ dB to validate the scaling of the received SNR as a function of $N$. On the plot, we mark the slope of the different curves, and as expected from  Theorem~\ref{thm:rate_characterization}, it is clear that while received SNR for a UE served by operator X scales as $N^2$, it also scales as $N$ for a UE served by operator~Y.
    		\begin{figure}[t]
    		\vspace{-0.4cm}
    		\centering
        		\includegraphics[width=\linewidth]{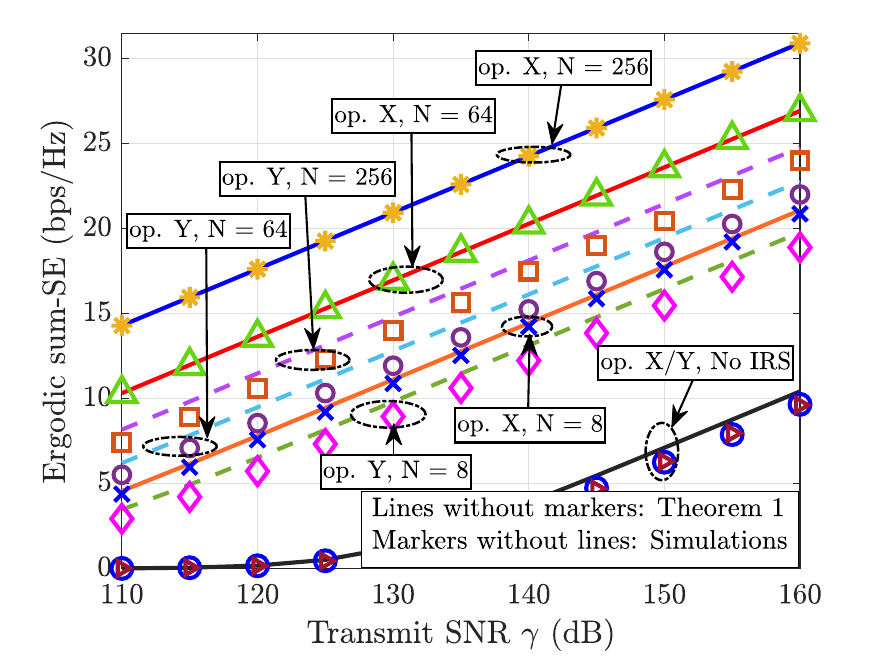}
		\caption{Ergodic sum-SE vs. transmit SNR.}
    		\label{fig:SE_txt_SNR}
    	 	\end{figure}
    		\begin{figure}[t]
    		\vspace{-0.4cm}
    		\centering
    		\includegraphics[width=\linewidth]{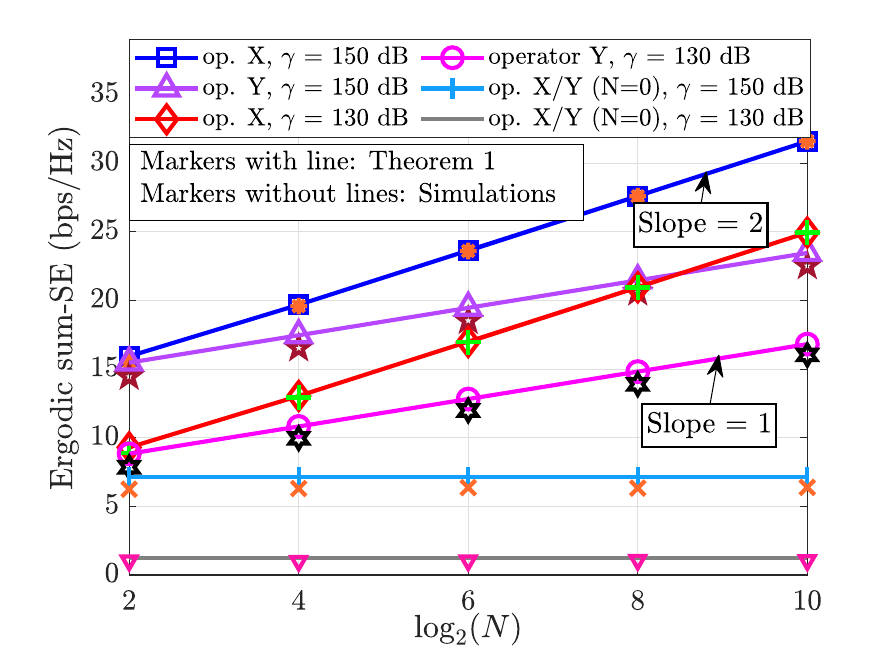}
    		\caption{Ergodic sum-SE vs. $\log_2 (N)$.}
    		\label{fig:SE_logN}
    		\vspace{-0.5cm}
    		\end{figure}
    	\begin{figure*}[t]
    		\vspace{-0.5cm}
    				\begin{subfigure}{0.33\linewidth}
    			\hspace{-0.1cm}
    			\includegraphics[width=1.12\linewidth]{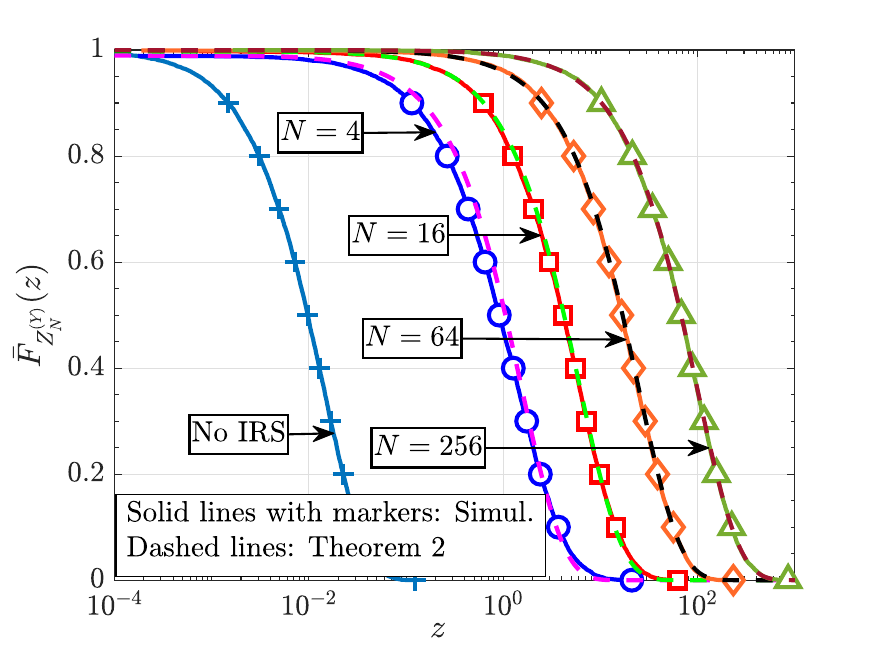}
    			\caption{CCDF of $Z^{(Y)}_N$ as a function of $N$.}
    			\label{fig:ccdf}
    	\end{subfigure}
    	\hspace{0.1cm}
    	\begin{subfigure}{0.33\linewidth}
    				\centering		
    				\includegraphics[width=1.12\linewidth]{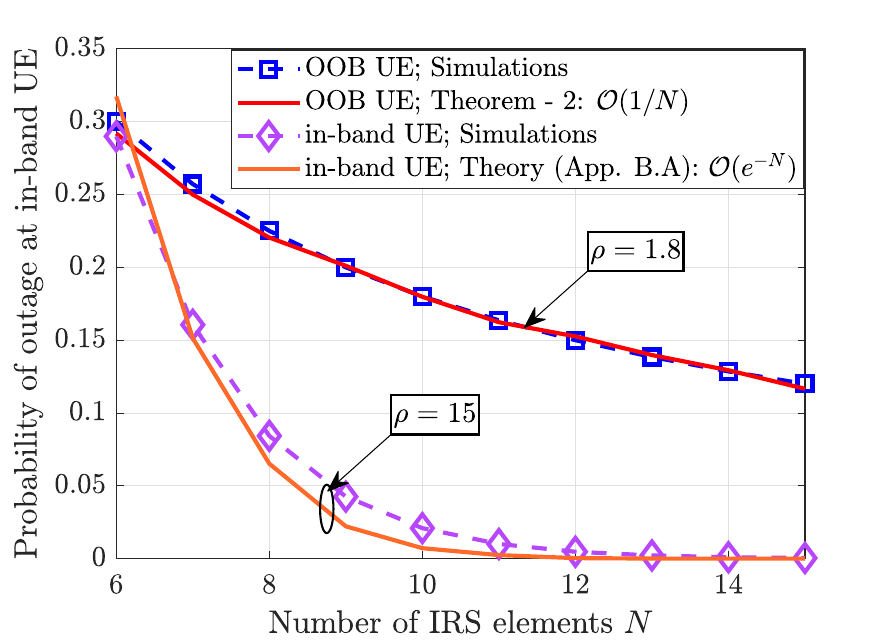}
    				\caption{Decay of outage prob. at in-band UEs vs. OOB UEs.}
    			\label{fig_outage_prob_inband_OOB}
    			\end{subfigure}
    			\hspace{0.01cm}
    			\begin{subfigure}{0.33\linewidth}
    			  \includegraphics[width=1.08\linewidth]{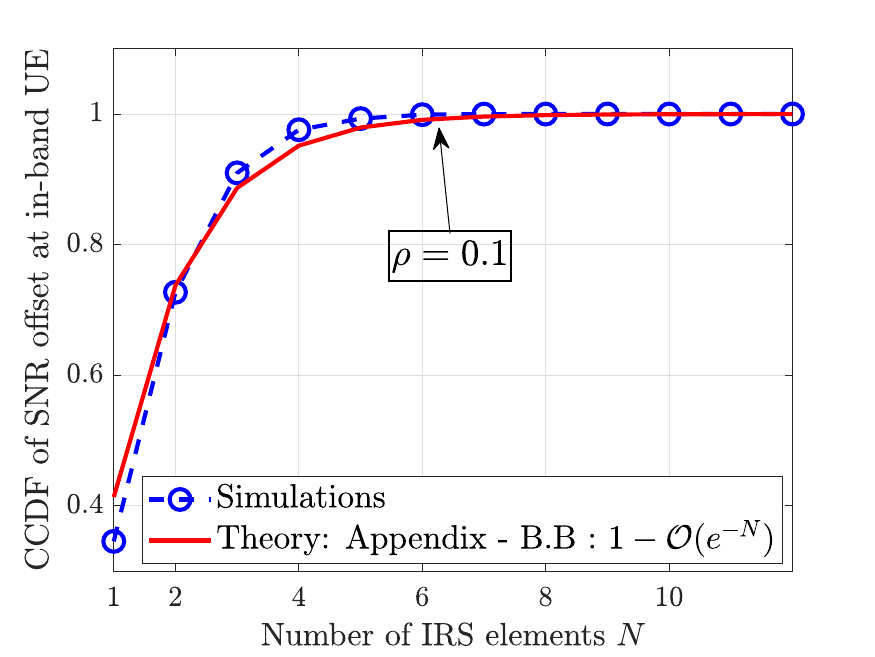}    						
			  \caption{CCDF of the SNR offset at in-band UEs.}
    			\label{fig_CCDF_offset_in_band}
    	\end{subfigure}
    			\caption{Instantaneous channel quality of the in-band/OOB UEs in the sub-6 GHz bands of communication.}
    			\vspace{-0.4cm}
    	\end{figure*}	
    	\begin{figure}[t]
    	\vspace{-0.1cm}
    	\centering
    	\includegraphics[width=\linewidth]{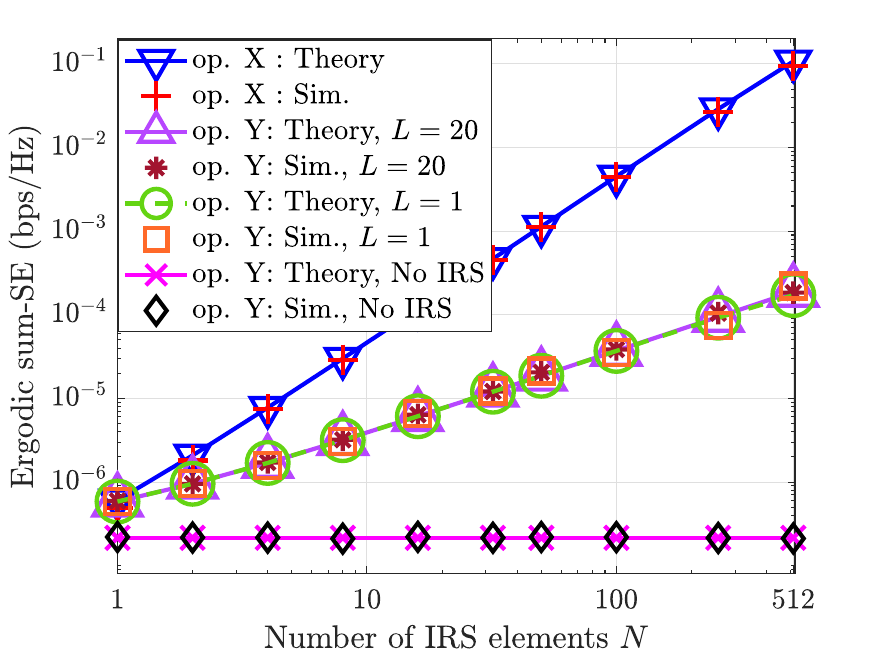}
    	\caption{Ergodic sum-SE vs. $N$ in LoS scenarios at $C_0\gamma=90$ dB.}    	\label{fig:SE_N_mmwave_single_path_low_SNR}
    \vspace{-0.1cm}
    \end{figure}
    \begin{figure}[t]
    	\vspace{-0.1cm}
    	\centering
    	\includegraphics[width=\linewidth]{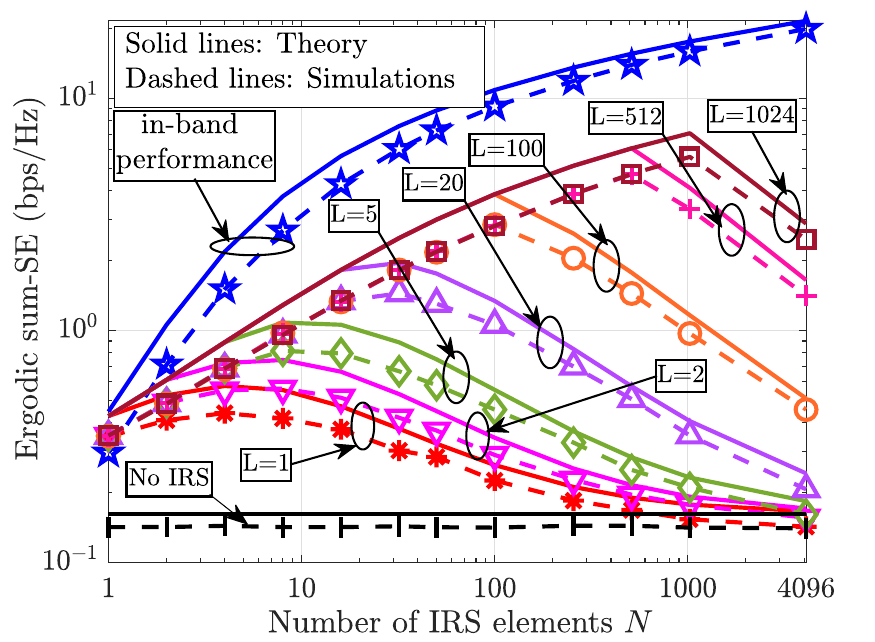}
    	\caption{Ergodic sum-SE vs. $N$ in LoS scenarios at $C_0\gamma = 150$ dB.}    	\label{fig:SE_N_mmwave_single_path_high_SNR}
    	\vspace{-0.25cm}
    \end{figure}
    \begin{figure*}[t]
    	\vspace{-0.5cm}
    	\begin{subfigure}{0.33\linewidth}
    		\hspace{-0.45cm}
    		\includegraphics[width=1.12\linewidth]{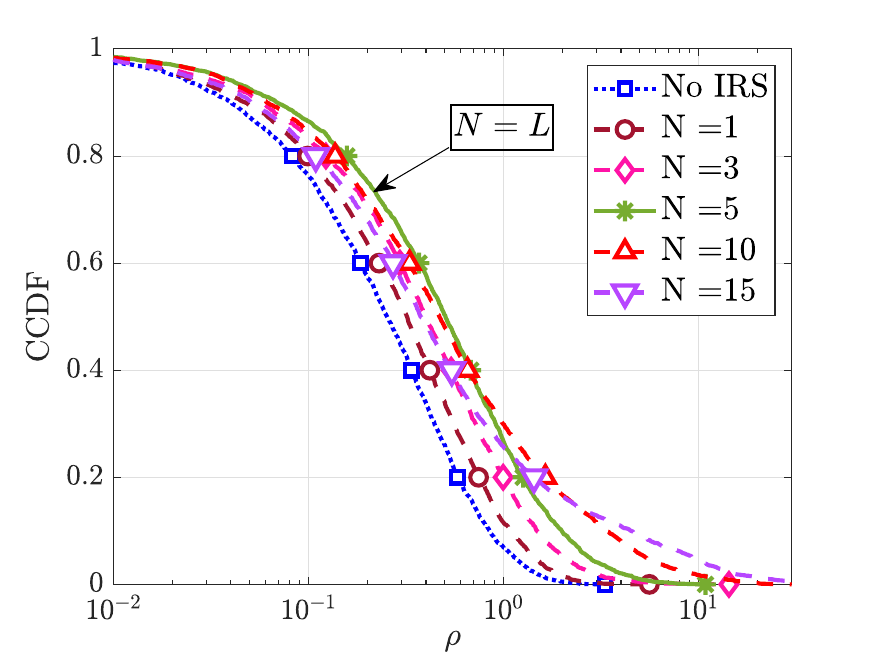}
    		\caption{$L=5$.}
    		\label{fig:CCDF_mmwave_single_path_L_5}
    	\end{subfigure}
    	\begin{subfigure}{0.33\linewidth}
    		\hspace{-0.4cm}
    		\includegraphics[width=1.12\linewidth]{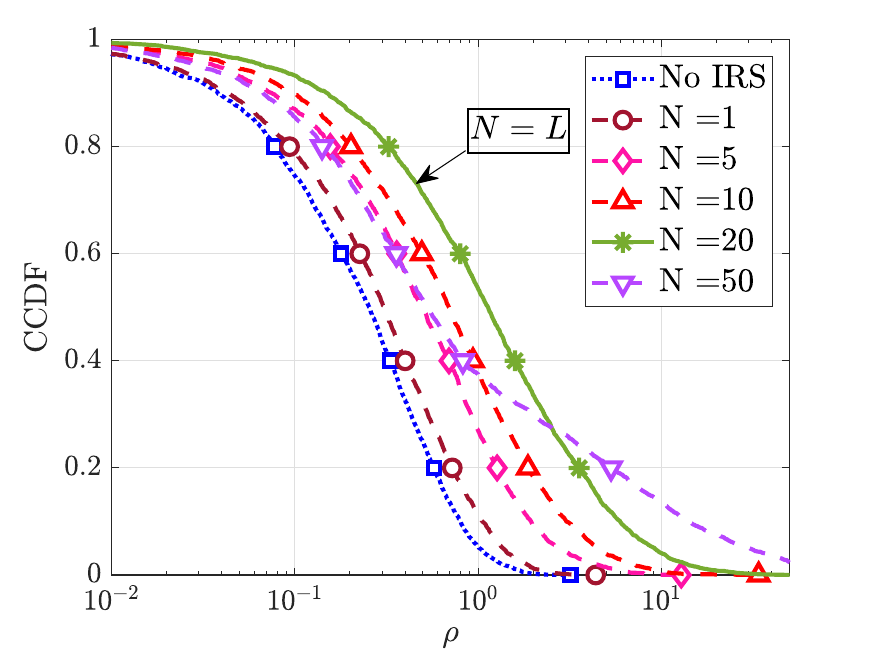}
    		\caption{$L=20$.}
    		\label{fig:CCDF_mmwave_single_path_L_20}
    	\end{subfigure}
    	\begin{subfigure}{0.33\linewidth}
    		\hspace{-0.22cm}
    		\includegraphics[width=1.12\linewidth]{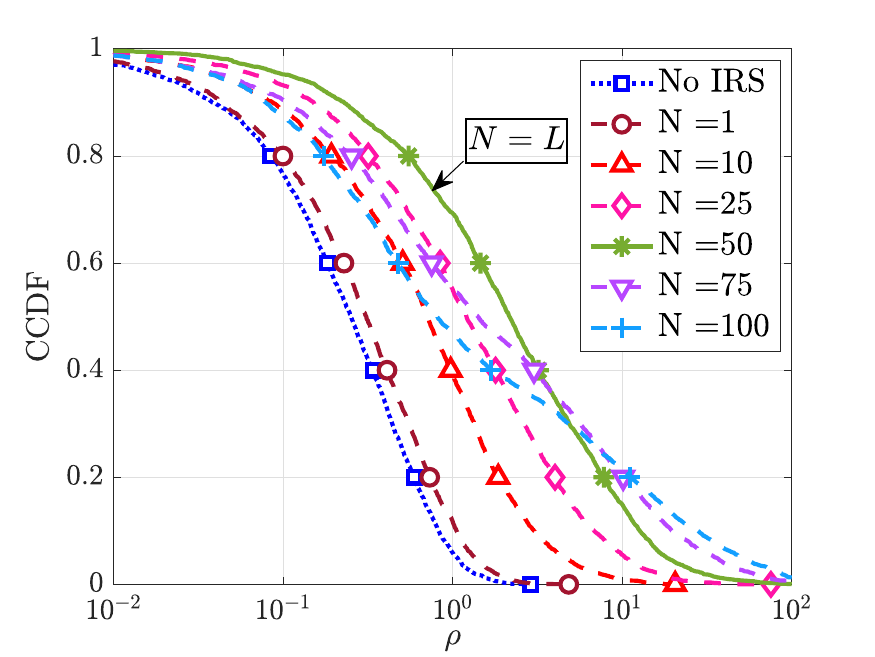}
    		\caption{$L=50$.}
    		\label{fig:CCDF_mmwave_single_path_L_50}
    	\end{subfigure}
    	\caption{CCDF of OOB UE's channel gain $|h_q|^2$ as a function of $N$ in LoS scenarios of the mmWave bands at $C_0\gamma = 150$ dB.}
    	\label{fig:overall_mmwave_ccdf}
    	\vspace{-0.5cm}
    \end{figure*}
    \begin{figure}[t]
    	\vspace{-0.05cm}
    	\centering
    	\includegraphics[width=\linewidth]{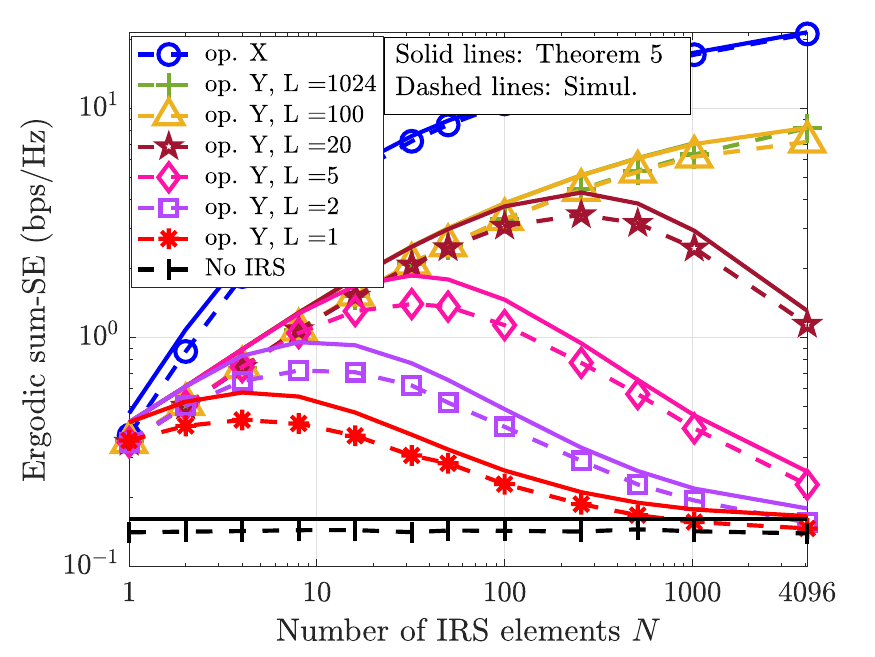}
    	\caption{Ergodic sum-SE vs. $N$ in (L+)NLoS scenarios.}
    	\label{fig:ergodic_SE_vs_N_multiple_path_high_SNR}
    	\vspace{-0.2cm}
    \end{figure}
    \begin{figure}[t]
    	\centering
    	\includegraphics[width=\linewidth]{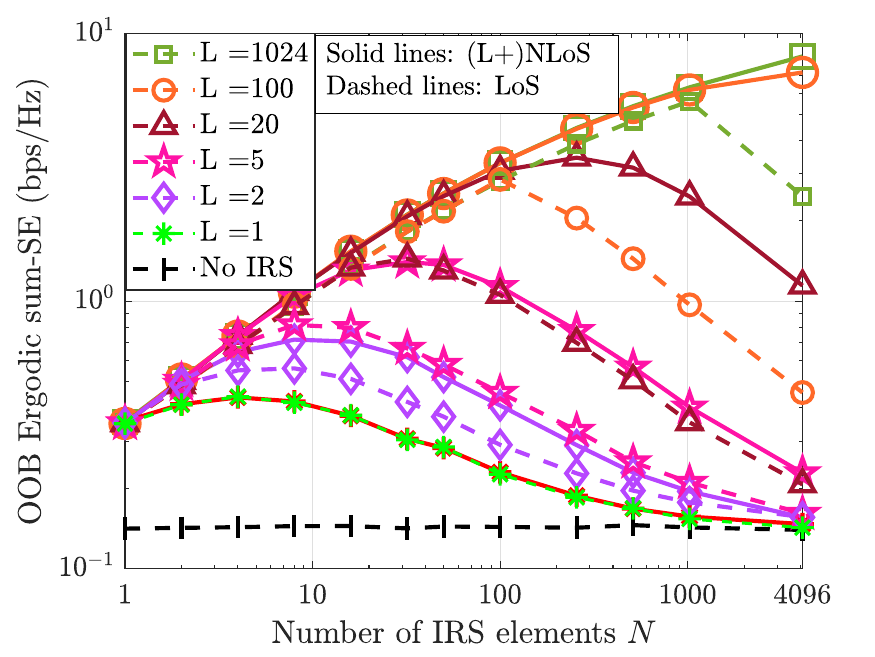}
    	\caption{Comparison of the sum-SE in LoS and (L+)NLoS scenarios.}
    	\label{fig:ergodic_SE_vs_N_LOS_vs_NLOS_high_SNR}
    	\vspace{-0.2cm}
    \end{figure}
    \begin{figure}[t]
    	\vspace{-0.2cm}
    	\centering
    	\includegraphics[width=1.02\linewidth]{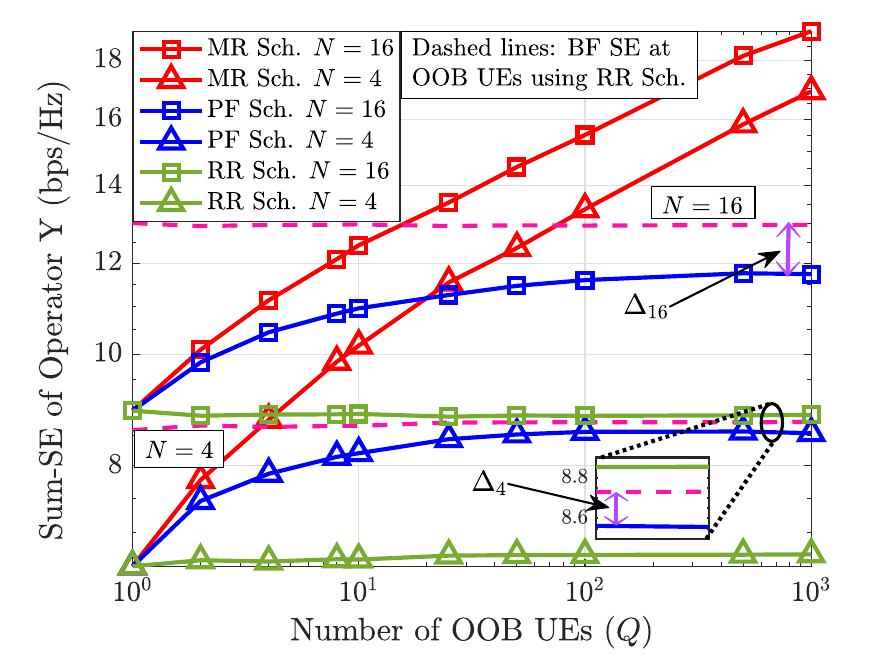}
    	\caption{OOB SE versus $Q$ for different schedulers.}
    	\label{fig:OOB_SE_vs_Q}
    	\vspace{-0.1cm}
    \end{figure}
    \begin{figure}[t]
        	\vspace{-0.05cm}
    	\centering
    	\includegraphics[width=\linewidth]{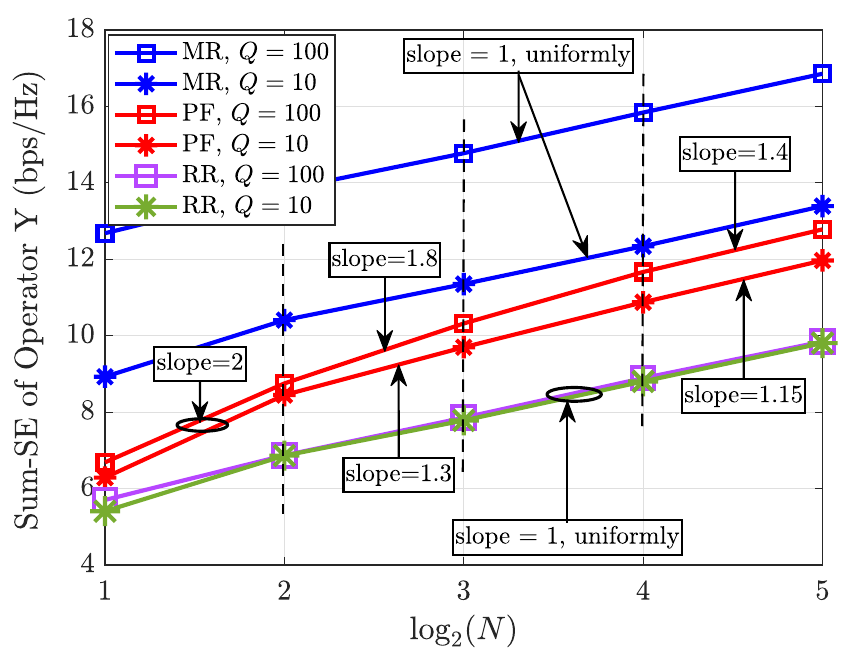}
    	\caption{OOB SE versus $\log_2(N)$ for different schedulers.}
    	\label{fig:OOB_SE_vs_N_diff_schedulers}
    	\vspace{-0.25cm}
    \end{figure}
    		 Next, we study the effect of the IRS on the OOB operator Y by considering the behavior of the random variable $Z^{(Y)}_N$ (see~\eqref{eq:snr_offset}),
    		which represents the difference in the SNR/channel gain at a UE $q$ served by BS-Y (which does not control the IRS) with and without the IRS in the environment. In Fig.~\ref{fig:ccdf}, we plot the empirical and theoretical CCDF of $Z^{(Y)}_N$ (at $\gamma=130$ dB.)
    		The analytical expression derived in Theorem~\ref{thm:exact_ccdf} matches well with the simulations. Also, almost surely, $Z^{(Y)}_N$ is a non-negative random variable for any $N>0$, which confirms that the channel gain at an OOB UE \emph{with} an IRS is at least as good as the channel gain at the same UE \emph{without} an IRS with probability $1$. The CCDF shifts to the right as the $N$ is increased, as expected. The left-most curve in the figure is the CCDF of received SNR in the absence of the IRS, which shows that the probability that an operator benefits from the presence of a randomly configured IRS increase with $N$ for operators who do not control the IRS. Thus, the instantaneous SNRs witnessed at an arbitrary UE of an OOB operator stochastically dominate the SNR seen by the same UE in the absence of the IRS, in line with Proposition~\ref{sec:prop_stochastic_dominance_sub_6_GHz}. \\
    		\indent Finally, we compare the outage probabilities of the in-band and OOB UEs in Fig.~\ref{fig_outage_prob_inband_OOB} as a function of $N$. While the probability decreases as $\mathcal{O}(e^{-N})$ at the in-band UEs (in-line with Remark~\ref{outage_prob_ccdf_in_band_UE}), it also uniformly decreases at OOB UEs, but at the rate of $\mathcal{O}(1/N)$ as per Theorem~\ref{thm:exact_ccdf}. Thus, the outage probability of the OOB UEs monotonically decreases with the IRS elements \emph{free of cost.} For completeness, in Fig.~\ref{fig_CCDF_offset_in_band}, we also validate that the CCDF of the SNR offset as in Remark~\ref{outage_prob_ccdf_in_band_UE} for the in-band UEs grows to $1$ at the rate of $1-\mathcal{O}(e^{-N})$.
    		\vspace{-0.1cm}
    \subsection{OOB Performance in mmWave Bands using RR Scheduler}
    We now numerically illustrate our findings in the mmWave band. First, we focus on the OOB performance when the IRS is optimized to in-band UEs in the LoS scenarios.
    In Figs.~\ref{fig:SE_N_mmwave_single_path_low_SNR}~and~\ref{fig:SE_N_mmwave_single_path_high_SNR}, we plot the ergodic sum-SEs vs.
    $N$, of both operators, at low and high (received) SNRs, respectively. 
    At low SNR, we observe that $1$) The ergodic sum-SE scales linearly with $N$, and $2$) The OOB performance is insensitive to $L$. This can be explained by analyzing the behavior of the SE in \eqref{eq:mmwave_rate_jio_rr_single_path} at low SNRs. Namely $\log(c\gamma) \approx c\gamma$ for small $c\gamma$. Thus, an OOB UE witnesses an SE which scales as $N^2/L$ for an $L/N$ fraction of time, leading to an effective scaling of SE as $\mathcal{O}(N)$.  Thus, if the OOB system is designed to operate at low SNRs, then even an OOB UE obtains benefits from the IRS that monotonically increase with $N$. On the other hand, at high SNRs, for any given $L$, the SE is an unimodal function of $N$, as described in Remark~\ref{sec:favourable_unfavourable_region_mmwave_single_path}. As long as $L \geq N$, i.e., the number of paths is sufficiently large (at least as many as the number of resolvable beams the IRS can form), almost surely, the IRS aligns with one of the OOB UE's channel angles and provides benefits which linearly increases with $N$. However, when $L\! <\! N$, the IRS beam does not always align with the OOB UE; hence, the performance starts to decline as $N$ grows with the peak SE obtained when $L=N$. In any case, the OOB SE is higher than the SE obtained without the IRS. This confirms that the IRS \emph{never} degrades the average OOB performance even in mmWave bands. Also, the simulations and analytical expressions match very well, which validates the correctness of Theorem~\ref{thm:rate_characterization_mmwave_single_path_IB}. Finally, we note that the SNR of the in-band operator monotonically grows as $N^2$, as expected. 
        	
    Next, we validate the instantaneous OOB performance by plotting the CCDFs of the channel gains of an arbitrary OOB UE in Figs.~\ref{fig:CCDF_mmwave_single_path_L_5}, \ref{fig:CCDF_mmwave_single_path_L_20}, and \ref{fig:CCDF_mmwave_single_path_L_50}, for $L=5$, $20$, and $50$, respectively. Clearly, the effective channel in the presence of the IRS stochastically dominates the effective channel in its absence, as found in Theorem~\ref{thm:ccdf_mmwave_OOB_single_path}. Further, with high probability, the channel gains of the OOB UEs achieve their maximum value when $N=L$, which is in line with our previous observations. 
    
    Next, we focus on the case when the IRS is optimized for (L+)NLoS mmWave scenarios. Specifically, we study the operators' ergodic sum-SE behavior in a high SNR regime. In Fig.~\ref{fig:ergodic_SE_vs_N_multiple_path_high_SNR}, we plot the ergodic sum-SE vs. $N$ when the IRS is optimized for (L+)NLoS scenarios. The OOB performance in this case is very similar to the LoS scenarios, where the enhancement in the OOB SE due to the IRS is contingent on the number of spatial paths available in the OOB UEs' channel. However, the main difference lies in the number of paths needed in the OOB UE's channel to obtain performance enhancement by the IRS.  While the LoS scenarios require $L$ to scale with $N$, it is sufficient for  $L$ to scale with $\sqrt{N}$ in the (L+)NLoS scenarios to obtain an OOB performance boost which is monotonic in $N$. Intuitively, this is because, in addition to the probability that one of the $L$ paths of the OOB channel aligns with the IRS (which scales as $L/N$), the IRS provides an additional opportunity for alignment by having a directional response in $L$ directions. This leads to the probability of alignment improving to $L^2/N$ in the OOB system. We confirm this in Fig.~\ref{fig:ergodic_SE_vs_N_multiple_path_high_SNR}: the peak performance is attained roughly at $N=L^2$ for every~$L$. 
    Recall that the probability that $i$ of the IRS beams align with the OOB channel is given by~\eqref{eq:prob_oob_NLOS}. Consider the scenario $L \ll N$. Using Stirling's approximation $\mybinom[1]{n}{i} \approx n^i/i!$ we can simplify~\eqref{eq:prob_oob_NLOS} as ${\sf {Pr}}(Y=i) \approx \frac{1}{i!}\left(\frac{L^2}{N}\right)^i $. Then, letting $i=1$ (to compare against the LoS scenario), the probability is $L^2/N$, supporting the fact that $\sqrt{N}$ paths are sufficient to obtain an OOB SE that log-linearly grows with the IRS elements. 
    Finally, Fig.~\ref{fig:ergodic_SE_vs_N_LOS_vs_NLOS_high_SNR} compares the OOB performance in LoS and (L+)NLoS scenarios. For a fixed $N$, $L$, the performance is enhanced much faster in (L+)NLoS than in LoS scenarios, in line with our discussions. 

    \subsection{OOB performance with PF and MR schedulers}
    We now numerically illustrate the OOB performance enhancement that can be obtained using the PF and MR schedulers described in Sec.~\ref{sec:PF_MR_schedulers}. 
    We plot the OOB SE versus the number of OOB UEs $Q$ in Fig.~\ref{fig:OOB_SE_vs_Q}.  Here, we compare the SE for RR, PF, and MR schedulers for different $N$. Using an MR scheduler, the OOB SE increases monotonically as $\mathcal{O}\left(\log(\log(Q))\right)$, in line with Lemma~\ref{lem:MR_scheduler_ergodic_SE}. Next, we consider the OOB SE using a PF scheduler. For $N=4$, the SE increases with $Q$, and it converges to the SE that is achievable as though the IRS is also optimal to the OOB UEs under RR scheduling, in line with Lemma~\ref{lem:PF_sub_6}. Thus, for large $Q$, OOB UEs enjoy optimal beamforming IRS benefits without explicitly optimizing the IRS to them when the OOB operator uses a PF scheduler. Even with $N=16$ using a PF scheduler, a similar trend as in $N=4$ is observed, except that the convergence rate to the optimal SE is now reduced. We observe this from the unequal gaps between the optimal SE and SE obtained by PF schedulers for $N=4$ and $16$, namely: $\Delta_{16}>\Delta_4$. This is because the degree of randomness increases with $N$, which requires the number of UEs to scale with $N$. We refer the reader to~\cite[Prop.~$1$]{Yashvanth_ICASSP_2023} for more details on the required number of UEs for a target gap in the SE. Finally, the OOB SE using the RR scheduler is invariant to $Q$ because it does not leverage multi-user diversity in the system.\\
    \indent In the final study of this section, we plot the OOB SE as a function of $\log_2(N)$ in Fig.~\ref{fig:OOB_SE_vs_N_diff_schedulers} for different schedulers. First, we observe that the slope of the curve using MR scheduler for both $Q=10,100$ is $1$, which validates the $\mathcal{O}(\log(N))$ dependence of the SE as in Lemma~\ref{lem:MR_scheduler_ergodic_SE}. Similarly, the RR scheduler also scales as $\mathcal{O}(\log_2(N))$, in line with Theorem~\ref{thm:rate_characterization}. On the other hand, the performance of the PF scheduler as a function of $N$ crucially depends on the relative value of $Q$ with respect to $N$. For smaller $N$, the OOB SNR scales with $N^2$, and as $N$ increases, the OOB SNR scales as $N^{\epsilon}$, where $\epsilon \leq 2$ (notice this in the decreasing values of slopes.) This happens because as $N$ increases, $Q$ is not large enough for the random IRS configuration to be near-optimal to at least one OOB UE at every point in time.
       \vspace{-0.1cm}
    		\section{Conclusions and Future Work}

    		In this paper, we studied a fundamental issue in an IRS-aided system, namely, the effect of deploying an IRS on the average and instantaneous performance of an OOB operator who has no control over the IRS. Surprisingly, we found that while the IRS optimally serves the in-band UEs, it simultaneously, and at no additional cost, enhances the quality of the channels of the OOB UEs compared to the system without any IRS. This performance enhancement is due to the reception of multiple copies of the signal at the OOB UEs through the IRS. Although this is always true in the sub-6 GHz bands, the degree of enhancement in the mmWave communications is determined by the number of spatial paths in the OOB UE's channel. When the number of paths in the channel is at least the number of resolvable beams at the IRS, the gain is monotonic in the number of IRS elements similar to sub-6 GHz bands and is marginal otherwise. Further, we showed that the systems where IRS is optimized to all the spatial paths of the in-band UE's channel outperform the ones where IRS is aligned only to the dominant path. This is due to the additional degrees of freedom offered by the IRS in the former, although it incurs additional BS-IRS signaling overhead at the in-band operator to configure the IRS optimally. Therefore, the deployment of an IRS benefits all the co-existing network operators, albeit to a lesser extent than the operator that has control over the IRS phase configuration. Further, it is possible to obtain much better benefits at the OOB UEs by using an opportunistic scheduler at the OOB operator.\\
    		\indent Future work can include the effect of multiple antennas at the BSs/UEs and study the diversity-multiplexing gain tradeoff at OOB UEs. Other practical considerations, like the non-availability of perfect channel state information, quantized IRS phase shift levels, frequency selectivity of channels, time \& frequency offsets, near-field effects, etc., can be accounted for to provide more insights. Another interesting line of study is to analyze the impact of IRSs on the OOB performance in interference-limited scenarios such as multi-cell environments.	
		
		  \begin{appendices}
    		\renewcommand{\thesectiondis}[2]{\Alph{section}:}
    		   \section{Proof of Theorem~\ref{thm:exact_ccdf}}\label{sec:ccdf_proof}
    		We recognize that $\tilde{h}_{1,q} \triangleq |{h}_{1,q}|^2 \sim \exp(1/\left(N\beta_{r,q}+\beta_{d,q}\right))$, and $\tilde{h}_{2,q} \triangleq |{h}_{2,q}|^2\sim \exp(1/\beta_{d,q})$.
    		  Now, $Z^{(Y)}_N$ is the difference between two exponential random variables. 	    		We first show that the correlation coefficient between $\tilde{h}_{1,q}$ and $\tilde{h}_{2,q}$  decays inversely with $N$. 
    		Recall that the correlation coefficient is defined $\rho_{12} \triangleq \mathbb{E}\left[(\tilde{h}_{1,q}-\mathbb{E}[\tilde{h}_{1,q}])(\tilde{h}_{2,q}-\mathbb{E}[\tilde{h}_{2,q}])\right]\bigg/\sigma_1\sigma_2$,
    	 where $\sigma_1^2$ and $\sigma_2^2$ are the variances of $\tilde{h}_{1,q}$ and $\tilde{h}_{2,q}$, respectively. We can verify that  $1)$ $\mu_1 \triangleq \mathbb{E}[\tilde{h}_{1,q}] = N\beta_{r,q} + \beta_{d,q}$, and $\mu_2 \triangleq \mathbb{E}[\tilde{h}_{2,q}] =  \beta_{d,q}$ $2)$ 
    			$\sigma_1^2 = {\left(N\beta_{r,q} + \beta_{d,q}\right)}^2$, and $\sigma_2^2 =  \beta_{d,q}^2$.
    	 Thus, we get 
    		$\rho_{12} = \left(\mathbb{E}\left[\tilde{h}_{1,q}\tilde{h}_{2,q}\right] - \left(N\beta_{r,q} + \beta_{d,q}\right)\beta_{d,q}\right)\bigg/ \left(N\beta_{r,q} + \beta_{d,q}\right)\beta_{d,q}.$
    	 	Using~\eqref{eq:aux_RV_defn}, it is easy to verify that $\mathbb{E}\left[\tilde{h}_{1,q}\tilde{h}_{2,q}\right] = N\beta_{r,q}\beta_{d,q} + 2\beta_{d,q}^2$. After simplification, we have
    		\begin{equation}\label{eq:rho_value}
    		\rho_{12} =1/\Big(1 + N\left(\beta_{r,q}/\beta_{d,q}\right)\Big),
    		\end{equation} 
    		which decays inversely with $N$.
    		We now use a result from~\cite[Eq.~4.24]{Simon_2002_ProbabilityGaussian} which characterizes the distribution of the difference of two dependent 
    		and non-identically distributed chi-square random variables and obtain the CDF of $Z^{(Y)}_N$ as
    		\begin{equation}\label{eq:cdf_simon}
    			F_{Z^{(Y)}_N}(z) = \left\{
    			\begin{array}{lr}
    				\dfrac{8}{\mu_1\mu_2(1-\rho_{12}^2)\gamma\alpha^-}e^{\left(\frac{\alpha^- z}{4}\right)}, & \text{if } z < 0,\\
    				1 - \dfrac{8}{\mu_1\mu_2(1-\rho_{12}^2)\gamma\alpha^+}e^{-\left(\frac{\alpha^+ z}{4}\right)}, & \text{if } z\geq 0,
    			\end{array}
    			\right.
    		\end{equation} 
    		where 
    		\begin{equation}\label{eq:cdf_simon_parameters}
    			\gamma\! = \!\dfrac{2\sqrt{{(\mu_2\! -\!\mu_1)}^2\!+4\mu_1\mu_2(1-\!\rho_{12}^2)}}{\mu_1\mu_2(1-\rho_{12}^2)};	\alpha^\pm\!=\! \gamma \pm \dfrac{2\left(\mu_2\!-\!\mu_1\right)}{\mu_1\mu_2(1-\!\rho_{12}^2)}.
    		\end{equation}  
    		For large $N$, using~\eqref{eq:rho_value} to let $\rho_{12} \rightarrow 0$ in~\eqref{eq:cdf_simon}, \eqref{eq:cdf_simon_parameters}, and recognizing that $\forall z$, $\bar{F}_{Z^{(Y)}_N}(z) = 1 - F_{Z^{(Y)}_N}(z)$, we obtain 
		    		\begin{equation} \label{eq:CCDF_Simon1}
    			\bar{F}_{Z^{(Y)}_N}(z) = \left\{
    			\begin{array}{lr}
    				1-\dfrac{\mu_2}{\mu_1+\mu_2}e^{\frac{z}{ \mu_2}}, & \text{if } z < 0, \\ 
    				\dfrac{\mu_1}{\mu_1+\mu_2}e^{-\frac{z}{ \mu_1}}, & \text{if } z\geq 0,
    			\end{array}
    			\right.
    		\end{equation} 
    		Substituting for $\mu_1$ and $\mu_2$ into \eqref{eq:CCDF_Simon1} completes the proof. \qed
    		    
    \section{Proof of the decay rate of outage probability and SNR offset at the in-band UEs}\label{sec:in-band_outage_CCDF_scaling_sketch}
       		We compute the decay rate of the outage probability, and SNR offset for the in-band UEs (as in Theorem~\ref{thm:exact_ccdf}) below.	
    		\subsection{Outage Probability}
    		Let $X \!\triangleq\! \left||h_d| + \sum_n |f_n| |g_n| \right|^2$ be the channel gain at an arbitrary in-band UE (dropping index $k$). Then, the outage probability is
    		\begin{equation}\label{eq:basic_out_eqn}
    			P^{\rho}_{\mathrm{out}} = \mathrm{Pr}(X \leq \rho).
    		\end{equation}		  
    		Since the true distribution of $X$ is complicated due to multiple products and summations, we find an upper bound on $P^{\rho}_{\mathrm{out}}$ as
    		\begin{equation*}
    			X = \left| |h_d| + \sum_n |f_n||g_n| \right|^2 \!\geq\! \left|\sum_n |f_n||g_n| \right|^2  \triangleq Y, \ \ \text{almost surely}.
    		\end{equation*} 
    		    		As a result, we have
    		\vspace{-0.2cm}
    		\begin{equation}\label{eq_Upper_bound_outage_prob}
    			P^{\rho}_{\mathrm{out}} = 	\mathrm{Pr}(X \leq \rho) \leq \mathrm{Pr}\left(\left|\sum\nolimits_n|f_n||g_n| \right|^2 \leq \rho\right). 
    		\end{equation} Now by Central Limit Theorem (CLT)~\cite{Qin_CL_2020}, we have 
    		\begin{equation}\label{eq:real_CLT}
    		\!\!\!Z \triangleq \sum_n |f_n||g_n| \ {\longrightarrow} \ \mathcal{N}\left( N\dfrac{\pi}{4}\sqrt{\beta_r},N\left(1-\dfrac{\pi^2}{16}\right)\beta_r \right). 
    		\end{equation}	
       	Using~\eqref{eq:real_CLT}, we can simplify~\eqref{eq_Upper_bound_outage_prob} as
    		\begin{align}
    				P^{\rho}_{\mathrm{out}} &\stackrel{(a)}{\leq} 1 - Q\left(\frac{\sqrt{\rho}}{c_1\sqrt{N}} - c_2\sqrt{N}\right) - Q\left(\frac{\sqrt{\rho}}{c_1\sqrt{N}} + c_2\sqrt{N}\right) \nonumber \\
    		&\stackrel{(b)}{\approx} 
    		2Q\left(c_2\sqrt{N}\right),
    			\label{eqn_decay_final_outage}
    		\end{align} 
    		where, in $(a)$, $c_1 \triangleq \sqrt{\left(1-\frac{\pi^2}{16}\right)\beta_r}$ and $c_2 \triangleq \frac{\pi}{\sqrt{16-\pi^2}}$, and $(b)$ follows as for fixed $\rho$, $\frac{\sqrt{\rho}}{c_1\sqrt{N}} \rightarrow 0$ for large $N$. Now, by applying the Chernoff bound to $Q(\cdot)$ in~\eqref{eqn_decay_final_outage}, we see that $Q(c_2\sqrt{N})\rightarrow 0$ at the rate of $\mathcal{O}(e^{-N})$. Hence, the  outage probability in~\eqref{eq:basic_out_eqn} also decays at least as fast as $\mathcal{O}(e^{-N})$. 
		    \vspace{-0.4cm}	
    \subsection{CCDF of the SNR offset}
    Let the channel gain with and without an IRS be $h_1 \triangleq \left||h_d| + \sum_n |f_n||g_n|\right|^2$ and $h_2 \triangleq |h_d|^2$.  To compute the CCDF of $O \triangleq h_1 - h_2$, using~\eqref{eq:real_CLT}, we notice\!\!
    \vspace{-0.2cm}
    \begin{equation}\label{eq_lower_bound_CCDF}
    \!\!\! \bar{F}_O(\rho) \triangleq \mathrm{Pr}\left(h_1 - h_2 \geq \rho \right) \geq \mathrm{Pr}\Bigg(\underbrace{\left|Z\right|^2 - |h_d|^2}_{\triangleq O'} \geq \rho \Bigg).
    \vspace{-0.2cm}
    \end{equation} 
    Now, we can simplify the lower bound of~\eqref{eq_lower_bound_CCDF} as 
      \begin{align}
    	&\!\!\!\!\mathrm{Pr}\left(|Z|^2 - |h_d|^2 \leq \rho \right)\! = \!\mathbb{E}_Z\left[\mathrm{Pr}\left(|z|^2-|h_d|^2 \leq \rho \big\vert Z=z\right)\right] \\ &\!\!\!\! = \mathbb{E}_Z\left[\mathrm{Pr}\left(|h_d|^2 \geq |z|^2 - \rho \big\vert Z=z\right)\right]  \stackrel{(a)}{=} \mathbb{E}_Z\Big[e^{-\left(\frac{|z|^2-\rho}{\beta_d}\right)}\Big]  \\ &\!\!\!\!
    	= \sqrt{\frac{8}{N\pi(16-\pi^2)\beta_r }} \!\int_{-\infty}^{\infty}\!\!e^{-\left(\frac{|z|^2-\rho}{\beta_d}+\frac{\left(z-N\frac{\pi}{4}\sqrt{\beta_r}\right)^2}{2N\left(1-\frac{\pi^2}{16}\right)\beta_r}\right)} \text{d}z,\!\!\!
	    \end{align}
	     where in $(a)$, we used the property of the exponentially distributed random variable $|h_d|^2$. To compute the above integral, we define $\alpha \triangleq 2N\left(1-\frac{\pi^2}{16}\right)\beta_r$, $\eta \triangleq N \frac{\pi}{4}\sqrt{\beta_r}$. 
    
    Then, by using standard integral solvers, we can show that 
        \begin{equation}
    \mathrm{Pr}\left(|Z|^2-|h_d|^2 \geq \rho \right)  = 1 - \dfrac{\sqrt{\beta_d}}{\sqrt{\beta_d+\alpha}}e^{-\left(\!\!\frac{\beta_d\eta^2-(\beta_d+\alpha)\rho}{\beta_d^2+\alpha\beta_d}\!\!\right)}.
    \end{equation}
    The above probability can be shown to be proportional to
    \begin{equation}\label{eq_final_CCDF_offset}
     1 - (1/\sqrt{N}) \times e^{-N} \times e^{\left(\frac{1/N}{1/N}\right)} \sim 1-\mathcal{O}\left(e^{-N}\right).
    \end{equation}
    Thus, the lower bound grows to $1$ exponentially in $N$, and hence the CCDF of the random variable $O$ also grows to $1$ at least as fast as the decay rate of $e^{-N}$ to $0$.     \qed
    \vspace{-0.3cm}

    		\section{Proof of Theorem~\ref{thm:ccdf_mmwave_OOB_single_path}}\label{app:proof_ccdf-single_path_mmwave}
		\vspace{-0.5cm}
    	Assume $L<N$. From~\eqref{eq:final_OOB_channel_mmwave_single_path}, we compute  probability $P = $
    			\begin{equation}
    			\!\!{\sf{Pr}}	\left(\left| h_{d,q} + \dfrac{N^2}{\sqrt{L}} \sum_{l=1}^{L} \gamma^{(1)}_{l,Y}\gamma^{(2)}_{l,q}\mathbf{\dot{a}}_N^H(\omega^l_{Y,q})\mathbf{\dot{a}}_N(\omega^1_{X,k})\right|^2\!\!<\rho\right)\!\!. \!\!
    			\end{equation}
    			By the total law of probability, we obtain	$P =  {\sf {Pr}} (|h_q|^2<\rho\rvert\mathcal{E}_1){\sf {Pr}}(\mathcal{E}_1)  + {\sf  {Pr}} (|h_q|^2<\rho\rvert\mathcal{E}_0){\sf {Pr}}(\mathcal{E}_0),$
    			where $\mathcal{E}_0$ and $\mathcal{E}_1$ are as defined in Sec.~\ref{app:sec:oob_ergoidc_rate-mmwave_single_path}. Thus, we can simplify $P$ as
    			\begin{equation}
    		\!\!\!	{\sf{Pr}}\!\left(\left| h_{d,q} \!+\! \dfrac{N}{\sqrt{L}}  \gamma^{(1)}_{l^*,Y}\gamma^{(2)}_{l^*,q}\right|^2\!\!<\!\rho\!\right)\!{\sf{Pr}}(\mathcal{E}_1)  + 	{\sf{Pr}}\!	\left(\left| h_{d,q}\right|^2\!<\!\rho\!\right)\! {\sf{Pr}}(\mathcal{E}_0).\!\!
    			\end{equation} 
        			Let $\!\left\{\!\tilde{\gamma}^{(1)}_{l^*,Y},{\tilde{\gamma}}^{(2)}_{l^*,q}\!\right\} \! \triangleq \! \sqrt{\!\frac{N}{\sqrt{L}}}\! \left\{\! \gamma^{(1)}_{l^*,Y},\!\gamma^{(2)}_{l^*,q}\!\right\}$. Then, by~\cite[Eq. 17]{Sudarshan_SPL_2019},  \!
    			\begin{equation}
    				\!\!{\sf{Pr}}	\!\left(\!\left| h_{d,q} \!+ \! \tilde{\gamma}^{(1)}_{l^*,Y}\tilde{\gamma}^{(2)}_{l^*,q}\right|^2\!\!<\!\rho\!\right)\!\! =1-\dfrac{Le^{\frac{L\beta_{d,q}}{N^2\beta_{r,q}} }}{N^2\beta_{r,q}} \hspace{0.1cm}\mathcal{I}_0\left(\rho;\beta_{d,q},\dfrac{N^2}{L}\beta_{r,q}\right)\!.\!\!\!
    			\end{equation}  Also, $	{\sf{Pr}}	(\left| h_{d,q}\right|^2<\rho) =1- e^{-\rho/\beta_{d,q}}$. Collecting these together completes the proof of~\eqref{eq:ccdf_mmwave_single_path_jio} when $L<N$. Similarly, we can prove for $L \geq N$. To obtain~\eqref{eq:stochastic_dominance_mmwave_single_path},
    	    	it suffices to show that 
    			\begin{multline}
    				 \forall L\lessgtr N, \ \ {\sf {Pr}}(G_1>\!\rho) - {\sf {Pr}}(G_0>\!\rho) \geq 0   \stackrel{(a)}{\iff} \!\!  \\  \frac{\bar{L}}{N}\!\left(\!\frac{\bar{L}e^{\frac{\bar{L}\beta_{d,q}}{N^2\beta_{r,q}} }}{N^2\beta_{r,q}} \mathcal{I}_0\left(\!\rho;\beta_{d,q},\frac{N^2}{\bar{L}}\beta_{r,q}\!\right) \!- \! e^{-\rho/\beta_{d,q}}\!\right) \geq 0, \label{eq:final_claim_ST_DOM_mmwave_single_path}
    			\end{multline} where $(a)$ is due to~\eqref{eq:ccdf_mmwave_single_path_jio}.  
    		 From~\cite{Harris_CAM_Bessel_function}, we know that $\mathcal{I}_0(\cdot)$ is related to the generalized upper incomplete Gamma function $\Gamma(\alpha,x;b)$ for appropriate $\alpha,x,b$, and hence we can simplify
    			\begin{equation}
    			\!\!\!\mathcal{I}_0\!\left(\!\rho;\beta_{d,q},\frac{N^2}{\bar{L}}\beta_{r,q}\right)\! = \frac{N^2\beta_{r,q}}{\bar{L}}\hspace{0.05cm}\Gamma\left(1,\frac{\bar{L}\beta_{d,q}}{N^2\beta_{r,q}};\frac{\bar{L}\rho}{N^2\beta_{r,q}}\!\right)\!.\!\!
    			\end{equation} Then, it suffices to show that $f(\rho) \triangleq\underbrace{e^{\frac{\bar{L}\beta_{d,q}}{N^2\beta_{r,q}}}\hspace{0.2cm}\Gamma\left(1,\frac{\bar{L}\beta_{d,q}}{N^2\beta_{r,q}};\frac{\bar{L}\rho}{N^2\beta_{r,q}}\right)}_{\triangleq f_1(\rho)} - \underbrace{e^{-\rho/\beta_{d,q}}}_{\triangleq f_2(\rho)} \!\geq \!0$. We first prove a few properties of $f(\rho)$. 
    			\emph{1) $f(0)=0$}: Clearly, $f_2(0)=1$. Now, $f_1(0) = e^{\frac{\bar{L}\beta_{d,q}}{N^2\beta_{r,q}}}\hspace{0.2cm}\Gamma\left(1,\frac{\bar{L}\beta_{d,q}}{N^2\beta_{r,q}};0\right)$. But, $\Gamma\left(1,x;0\right) = \Gamma(1,x)$ where $\Gamma(\alpha,x)$ is the upper incomplete Gamma function. Hence, $f_1(0) = e^{\frac{\bar{L}\beta_{d,q}}{N^2\beta_{r,q}}}\hspace{0.2cm}\!\Gamma\left(1,\frac{\bar{L}\beta_{d,q}}{N^2\beta_{r,q}}\right)$.  But from~\cite[Eq. 8.4.5]{NIST_Math_handbook}, $\Gamma\left(1,x\right) \!=\! e^{-x}$. Thus, $f_1(0) \!= \!1$ and hence $f(0)\!=\!f_1(0)\!-\!f_2(0)=0$.\!
    				\begin{figure}[t]
    					\centering
    				\includegraphics[width=0.9\linewidth]{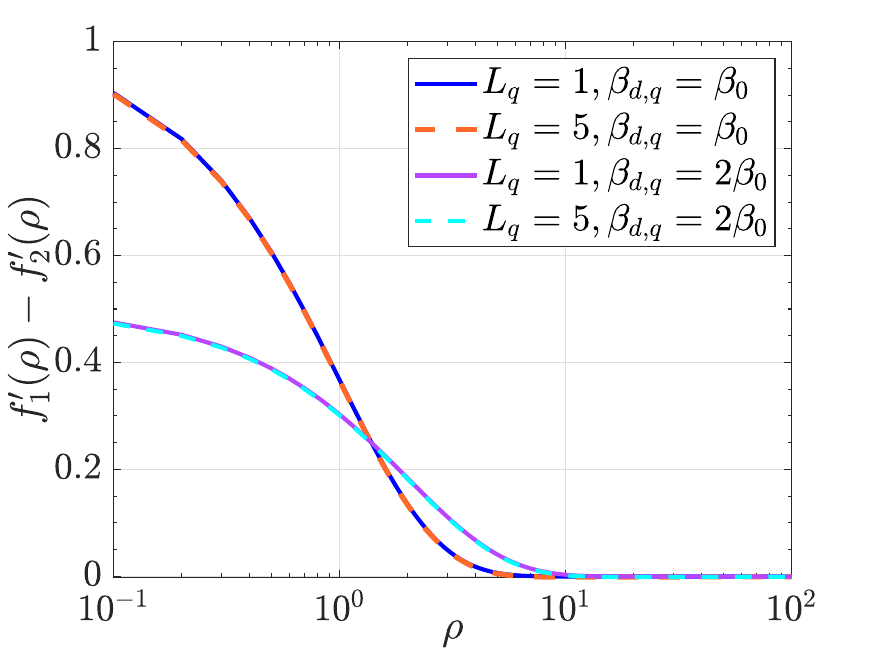} 
      				\caption{Graph of $f'(\rho) $ vs. $\rho$ with $N=128$.}
    			\label{fig:single_path_ccdf_analytical_check}
    			\end{figure}
    			\emph{2) $f(\rho)$ is non-decreasing}: From~\cite{Chaudhry_CAM_gamma_functions}, $\partial \Gamma(\alpha,x;b)/\partial b \! = - \Gamma(\alpha\!-1,x;b)$ and thus, $f_1'(\rho)=-e^{\frac{\bar{L}\beta_{d,q}}{N^2\beta_{r,q}}}\hspace{0.2cm}\Gamma\left(0,\frac{\bar{L}\beta_{d,q}}{N^2\beta_{r,q}};\frac{\bar{L}\rho}{N^2\beta_{r,q}}\right)$. Further, for practical values of the system parameters, $\bar{L}\beta_{d,q}/N^2\beta_{r,q} \!\leq\! 1$, and thus we can approximate $f_1'(\rho) \approx -2e^{\frac{\bar{L}\beta_{d,q}}{N^2\beta_{r,q}}}K_0\left(\frac{2}{N}\sqrt{\frac{\bar{L}\rho}{\beta_{r,q}}}\right)$, where $K_0 (x)$ is the zeroth-order modified Bessel function of the second kind. Further,  $f_2'(\rho) = -\left(1/d\right)e^{-\frac{\rho}{d}}$. Hence, we can show that $f'(\rho) = f'_1(\rho)-f'_2(\rho)\geq 0$ $\forall \rho \in \mathbb{R}^+$ (see  Fig.~\ref{fig:single_path_ccdf_analytical_check}).  Finally, since $f(\rho)$ is non-decreasing, $f(\rho)\geq f(0) = 0$. \qed
    		
    		\section{Proof of Lemma~\ref{lemma_correlation_function_IRS_mmwave_multiple_path}}\label{app_proof_correlation_function}

    		Using~\eqref{eq:optimal_IRS_mmwave_multi_path}, we simplify  $\rho_{\nu,\theta} = N\mathbb{E}\left[\left|\mathbf{\dot{a}}^H(\nu) \boldsymbol{\theta}^{\mathrm{opt}}\right|\right] $ as
    			\begin{align}
    		    				\!\!\!\!&\!\!\!= \mathbb{E}\left[\left|\sum\nolimits_{n=1}^{N}e^{j\pi(n-1)\nu} \dfrac{\sum\nolimits_{l=1}^{L} \gamma^{(1)*}_{l,X}\gamma^{(2)*}_{l,k}e^{-j(n-1)\pi\omega^l_{X,k}}}{\left|\sum\nolimits_{l=1}^{L} \gamma^{(1)*}_{l,X}\gamma^{(2)*}_{l,k}e^{-j(n-1)\pi\omega^l_{X,k}}\right|}\right|\right]\!\! \nonumber \\
    				\!\!\!\!&\!\!\! \! \stackrel{(a)}{\geq}\! \dfrac{1}{\sqrt{L}} \mathbb{E}\left[\!\dfrac{1}{\|\mathbf{h}_c\|_2}\!\left|\sum_{n=1}^{N}\!e^{j\pi(n-1)\nu}\sum\limits_{l=1}^{L}\! \gamma^{(1)*}_{l,X}\!\gamma^{(2)*}_{l,k}\!e^{-j(n-1)\pi\omega^l_{X,k}}\!\right|\!\right]\!\!\!  \nonumber\\
    				\!\!\!\!&\!\!\!\! \stackrel{(b)}{=} \dfrac{1}{\sqrt{L}}\hspace{0.1cm}\mathbb{E}\left[\dfrac{1}{\|\mathbf{h}_c\|_2}\left|\sum\limits_{l=1}^{L} \gamma^{(1)*}_{l,X}\gamma^{(2)*}_{l,k}\sum\limits_{n=1}^{N}e^{j\pi(n-1)(\nu-\omega^l_{X,k})}\right|\right] \nonumber\\
    				\!\!\!\!&\!\!\! = \dfrac{1}{\sqrt{L}}\hspace{0.1cm}\mathbb{E}\left[\dfrac{1}{\|\mathbf{h}_c\|_2}\left|\sum\nolimits_{l=1}^{L} \gamma^{(1)*}_{l,X}\gamma^{(2)*}_{l,k}\hspace{0.1cm}\dfrac{e^{jN\frac{\pi}{2}(\nu-\omega^l_{X,k})}}{e^{j\frac{\pi}{2}(\nu-\omega^l_{X,k})}} \right.\right. \nonumber\\[-0.05in]
    				&\hspace{3.3cm}\left. \left. \times\dfrac{\sin\left(0.5N\pi(\nu-\omega^l_{X,k})\right)}{\sin\left(0.5\pi(\nu-\omega^l_{X,k})\right)}\right|\right] \nonumber,\\[-0.25in] 
    			\end{align} 
      			where $\mathbf{h}_c \triangleq [\gamma^{(1)}_{1,X}\gamma^{(2)}_{1,k},\ldots,\gamma^{(1)}_{L,X}\gamma^{(2)}_{L,k}]^T$. Also, $(a)$ is because  $\left|\sum_{l=1}^{L} \gamma^{(1)*}_{l,X}\gamma^{(2)*}_{l,k}e^{-j(n-1)\pi\omega^l_{X,k}}\right| \leq \sqrt{L\sum_{l=1}^L|\gamma^{(1)}_{l,X}\gamma^{(2)}_{l,k}|^2}$ (the Cauchy-Schwarz inequality), and $(b)$ is obtained by changing the order of summation. We recognize that $ \!\sin\left(\!0.5N\pi(\nu\!-\!\omega^l_{X,k})\!\right)\!/\sin\left(\!0.5\pi(\nu\!-\!\omega^l_{X,k})\!\right) \!=\! F_N(\nu\!-\!\omega^l_{X,k})$ where $F_N(\cdot)$ is the \emph{Fej\'{e}r Kernel} given by 
			    	\begin{equation}
    				F_N(x) = \begin{cases}
    					N + o(N), &  \mathrm{ if } \ \ x = 0,\\
    					o(N),  &  \mathrm{ if } \ \ x = \pm\frac{k}{N}, k \in \mathbb{N}.
    				\end{cases}
    			\end{equation}
    			As a consequence, we have
    			\begin{equation}
    			\!\!\!	F_N(\nu-\omega^l_{X,k}) =\begin{cases}\!
    					N + o(N), &  \mathrm{ if } \ \nu \in \left\{\omega^1_{X,k}, \ldots,\omega^L_{X,k} \right\},\\ \!
    					o(N), & \mathrm{ if } \ \nu \in \mathbf{\Phi} \setminus \left\{\omega^1_{X,k}, \ldots,\omega^L_{X,k} \right\}.
    				\end{cases} 
    			\end{equation}
    			Thus, when $\nu \in \mathbf{\Phi}\setminus \left\{\omega^1_{X,k}, \omega^2_{X,k},\ldots,\omega^L_{X,k} \right\}, \rho_{\nu,\theta} \stackrel{N\rightarrow\infty}{\longrightarrow} 0$; and when $\nu \in \left\{\omega^1_{X,k}, \omega^2_{X,k},\ldots,\omega^L_{X,k} \right\}$, we have $\rho_{\nu,\theta} \geq \dfrac{N}{\sqrt{L}}\times\mathbb{E}\left[\dfrac{|h_{l^*}|}{\|\mathbf{h}_c\|_2}\right]$, implying $ \rho_{\nu,\theta} =
    			  \Omega\left(N/\sqrt{L}\right) + o(N)$. \qed
			      		\end{appendices}
\bibliographystyle{IEEEtran}
\bibliography{IEEEabrv,IRS_ref_short}
   
   \end{document}